\documentclass[11pt,letterpaper,oneside,english]{article}
\usepackage[left=3.5cm,right=3.5cm,top=3cm,bottom=3cm]{geometry}

\usepackage{amssymb,amsbsy,latexsym}
\usepackage{amsmath}
\usepackage{graphics, subfigure, float}
\usepackage[T1]{fontenc} % Important: Use before loading the fonts
\makeatletter
\let\@font@warningori\@font@warning
\newcommand\shutup{\def\@font@warning##1{}}
\newcommand\youcanspeak{\let\@font@warning\@font@warningori}
\shutup
\usepackage{fourier}
\youcanspeak
\usepackage[scaled=0.875]{helvet}

\usepackage{amscd,amsthm}
\usepackage{verbatim, comment}
\usepackage{fp,calc}
\usepackage{pst-all}
\usepackage{bm} % bold math font

\newtheoremstyle{theorem}{1em}{1em}{\slshape}{0pt}{\bfseries}{.}{ }{}
\theoremstyle{theorem}
\newtheorem{theorem}{Theorem}
\newtheorem*{theorem*}{Theorem}
\newtheorem{corollary}[theorem]{Corollary}
\newtheorem{lemma}[theorem]{Lemma}

\theoremstyle{remark}

\newtheorem*{remark*}{Remark}

\providecommand{\setN}{\mathbb{N}}
\providecommand{\setZ}{\mathbb{Z}}
\providecommand{\setQ}{\mathbb{Q}}
\providecommand{\setR}{\mathbb{R}}
\newcommand{\conv}[1]{\textrm{conv}(#1)} 
\newcommand{\cone}{\textrm{cone}} 
\newcommand{\supp}{\textrm{supp}}
\newcommand{\vertices}{\textrm{vert}}
\newcommand{\sign}{\textrm{sign}}

\newcommand{\enc}{\textrm{enc}}
\makeatletter

\usepackage[displaymath,textmath,sections,graphics, subfigure, floats]{preview} 
\PreviewEnvironment{myproblem}
\PreviewEnvironment{defn} 
\PreviewEnvironment{center} 
\PreviewEnvironment{pspicture} 
\DeclareMathAlphabet{\pazocal}{OMS}{zplm}{m}{n}

\makeatother

\def\parallelepiped{\pi}  % symbol for parallelepiped (was \Pi, now it's \pi)
\def\parallelepipedset{\Pi} % symbol for the set of parallelepiped (was \bm{\Pi}, now it's \Pi)

\begin{document}
%\begin{comment}
\title{Polynomiality for Bin Packing with a Constant Number of Item Types\thanks{A conference version of this paper appeared at SODA'14.}} %in fixed dimension} 

\author{Michel X. Goemans\thanks{MIT, Email: {\tt goemans@math.mit.edu}. Supported by ONR grants N00014-11-1-0053, N00014-17-1-2177 and NSF contract 1115849.} \and Thomas Rothvoss\thanks{University of Washington, Seattle. Email: {\tt rothvoss@uw.edu}. Supported by an Alfred P. Sloan Research Fellowship; a David and Lucile Packard Fellowship and NSF CAREER Award 1651861. Research was done while the 2nd author was postdoc at MIT.}}%\thanks{MIT, Email: {\tt rothvoss@math.mit.edu}}} % \\{ \tt{rothvoss@math.mit.edu}}}

\maketitle

\begin{abstract}
We consider the bin packing problem with $d$ different item sizes $s_i$ and 
item multiplicities $a_i$, where all numbers are given in binary encoding. 
This problem formulation is also known as the \emph{1-dimensional cutting stock problem}.
%Several authors have asked whether or not this problem admits a polynomial
%time algorithm if $d$ is constant. %In fact, one paper even conjectures it to
%be $\mathbf{NP}$-hard 

In this work, we provide an algorithm which, for constant $d$, solves bin packing 
in polynomial time. % even if all the data is binary encoded. 
This was an open problem for all $d\geq3$. % as stated by several authors.

In fact, for constant $d$ our algorithm solves the following problem in polynomial time:
given two $d$-dimensional polytopes $P$ and $Q$, find the smallest number of integer points
in $P$ whose sum lies in $Q$.

Our approach also applies to \emph{high multiplicity} scheduling problems
in which the number of copies of each job type is given in binary encoding and each 
type comes with certain parameters such as release dates, processing times and deadlines.
We show that a variety of high multiplicity scheduling problems can be solved in 
polynomial time if the number of job types is constant.
%and potentially exponentially large.
%in which we are given a constant number of job types, but potentially exponentially many 
%copies of each type. 
%each type with certain characteristics
%The main technical contribution lies in a surprising structure theorem 
%for conic integer combinations.
%For example for bin packing it implies that a polynomial number of special patterns can be 
%pre-computed such that obliviously for \emph{every} multiplicity vector $a$, almost all of 
%the weight in an optimum solution would be distributed on those patterns.
\end{abstract}

\section{Introduction}

Let $(s,a)$ be an instance for \emph{bin packing} with \emph{item sizes}
$s_1,\ldots,s_d \in [0,1]$ and a vector $a \in \setZ^d_{\geq0}$ of \emph{item multiplicities}.
In other words, our instance contains $a_i$ many copies of an item of size $s_i$.
In the following we assume that $s_i$ is given as a rational number and $\Delta$
is the largest number appearing in the denominator of $s_i$ or the multiplicities $a_i$.
Let $\pazocal{P} := \{ x \in \setZ_{\geq0}^d \mid s^Tx \leq 1\}$.
Now the goal is to select a minimum number of vectors from $\pazocal{P}$ that sum up to $a$, i.e.
\begin{equation} \label{eq:BinPacking}
  \min\Big\{ \bm{1}^T\lambda \mid \sum_{x \in \pazocal{P}} \lambda_x \cdot x = a; \; \; \lambda \in \setZ_{\geq0}^{\pazocal{P}} \Big\}
\end{equation}
where $\lambda_x$ is the \emph{weight} that is given to $x \in \pazocal{P}$. This problem is also known
as the \emph{(1-dimensional) cutting stock problem} and its study goes back to the classical paper
by Gilmore and Gomory~\cite{Gilmore-Gomory61}.
Note that even for fixed dimension $d$, the problem is that both, the number of points $|\pazocal{P}|$
and the weights $\lambda_x$ will be exponentially large.
Let $OPT$ and $OPT_f$ be the optimum integral and fractional solution to \eqref{eq:BinPacking}.
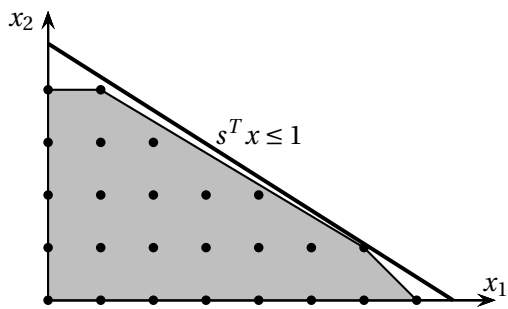
\begin{figure}
\begin{center}
\psset{unit=0.7cm}
\begin{pspicture}(-1,-0.2)(10,6)
\psaxes[ticks=none,labels=none,arrowsize=5pt]{->}(0,0)(0,0)(8.5,5.5)
\pspolygon[fillstyle=solid,fillcolor=lightgray](0,0)(0,4)(1,4)(6,1)(7,0)
\def\sI{0.13}
\def\sII{0.205}
\FPeval{\sIrec}{1.0 / \sI}
\FPtrunc{\sIrec}{\sIrec}{4}
\FPeval{\sIIrec}{1.0 / \sII}
\FPtrunc{\sIIrec}{\sIIrec}{4}
\psline[linewidth=1.5pt](\sIrec,0)(0,\sIIrec)
\FPtrunc{\maxI}{\sIrec}{0}
\FPtrunc{\maxII}{\sIIrec}{0}
\FPeval{\maxIp}{\maxI + 1} \FPtrunc{\maxIp}{\maxIp}{0}
\FPeval{\maxIIp}{\maxII + 1} \FPtrunc{\maxIIp}{\maxIIp}{0}
\multido{\I=0+1}{\maxIp}{ % x-Koordinate
 \FPeval{\iter}{((1 - \I * \sI) / \sII) + 1}  % maximales ganzzahliges y+1 so dass (x,y) im Knapsack Polytop fuer x=\I
 \FPtrunc{\iter}{\iter}{0}
 \multido{\N=0+1}{\iter}{ % y-Koordinate
   \psdot(\I,\N)
 }
}
\rput[c](8.5,5pt){$x_1$}
\rput[r](-5pt,5.3){$x_2$}
\rput[c](4,3.2){$s^Tx \leq 1$}
\end{pspicture}
\caption{Knapsack polytope for $s = (0.13, 0.205)$.}
\end{center}
\end{figure}
As bin packing for general $d$ is strongly $\mathbf{NP}$-hard~\cite{Johnson1992},
we are particularly interested in the complexity of bin packing if $d$ is constant.
For $d=2$ it is true that $OPT = \lceil OPT_f\rceil$ and it suffices to compute and round an optimum 
fractional solution~\cite{BinPackingD2-McCormickSmallwoodSpieksmaSODA97}. 
However, for $d\geq3$, one might have $OPT > \lceil OPT_f\rceil$. Still, \cite{BinPackingWithDminus2Bins-FilippiAgnetis05} generalized the argument of \cite{BinPackingD2-McCormickSmallwoodSpieksmaSODA97}
to find a solution with at most $d-2$ bins more than the optimum in polynomial time.

%The framework of De Loera, Hemmecke, Onn and Weismantel~\cite{NfoldIntegerProgramming-MOR2008} on Graver basis
%can find an optimum solution in time polynomial in the total size $s^Ta$, for fixed $d$\footnote{As in this case, the values $\lambda_x$ are polynomially bounded, using the insights of Eisenbrand and Shmonin~\cite{IntegerConesEisenbrandShmoninORL06} on the support of an optimum solution, one could also solve the problem with an integer linear program with a constant number of variables.}.
The best polynomial time algorithm previously known for constant $d\geq3$ is an $OPT+1$ approximation algorithm by Jansen and Solis-Oba~\cite{DBLP:conf/ipco/JansenS10}
which runs in time $2^{2^{O(d)}} \cdot (\log \Delta)^{O(1)}$.
Their algorithm uses a \emph{mixed integer linear program} and is based on the following insights: 
(1) If all items are small, say $s_i \leq \frac{1}{2d}$,
then the integrality gap is at most one\footnote{Compute a basic solution $\lambda$ to the LP and buy $\lfloor\lambda_x\rfloor$ times point $x$. Then assign the
items in the remaining instance greedily.}.
(2) If all items have constant size, then one can guess the points used in the optimum solution.
It turns out that for arbitrary instances both approaches can be combined for an $OPT+1$ algorithm. Also note that if the \emph{mixed integer roundup conjecture} holds true, then there is indeed a significantly simpler algorithm achieving the same bound, again by Jansen and Solis-Oba~\cite{DBLP:journals/ipl/JansenS11}.
However, to find an optimum solution, we cannot allow any error and a fundamentally different approach is
needed.

Note that for general $d$, the recent algorithm of Hoberg and the 2nd author 
provides solutions of cost at most $OPT + O(\log d)$~\cite{BinPacking-log-OPT-LogLog-OPT-Rothvoss-FOCS2013,DBLP:conf/soda/HobergR17}, improving
the classical Karmarkar-Karp algorithm with a guarantee of $OPT + O(\log^2 d)$~\cite{KarmarkarKarp82}.
Both algorithms run in time polynomial in $\sum_{i=1}^d a_i$ and thus count in our setting
as pseudopolynomial. In fact, those algorithms can still be cast as asymptotic
FPTAS.

Bin packing and more generally the cutting stock problem belong to a family of problems that 
consist of selecting integer points in a polytope with multiplicities. In fact, several 
scheduling problems fall into this framework as well, where the polytope describes the 
set of jobs that are admissible on a machine under various constraints.

We give some notation needed throughout the paper. For a set $X \subseteq \setR^d$, we define the spanned \emph{cone} as $\textrm{cone}(X) = \{ \sum_{x \in X} \lambda_x x \mid \lambda_x \geq 0  \; \forall x \in X \}$ and the \emph{integer cone} as
 $\textrm{int.cone}(X) := \{ \sum_{x \in X} \lambda_x x \mid \lambda_x \in \setZ_{\geq0} \; \forall x \in X \}$. 
 We will only consider rational polyhedra $P = \{ x \in \setR^d \mid Ax \leq b\}$ where both $A$
 and $b$ are rational. %Note that this is equivalent to requiring that $P$ is the
% convex hull of a finite number of points with rational coordinates.
 % Note that for any rational polytope in inequality representation
% one can multiply each inequality $A_ix \leq b_i$ to achieve integrality.
 We define $\textrm{enc}(P)$ as the number of bits that
 it takes to write down the rational numbers in $A$ and in $b$.
 For definitness one can use the definitions of Korte and Vygen~\cite{KorteVygen02}, Chapter 4.
 However, we will only need that $\textrm{enc}(P)$ is polynomially
 related to $\max\{ m,d, \log \Delta\}$ where $m$ is the number of inequalities and $\Delta$ is the largest
 nominator or denominator
 appearing in any of the inequalities definining of $P$. For the remainder of this paper we denote $\log(y) := \log_2(y)$.
 We would like to point out that any rational inequalities
 can be multiplied by a least common multiple of the denominators to obtain
 an equivalent system $\tilde{A}x \leq \tilde{b}$ with integral coefficients whose encoding length is within a
 polynomial factor of the original one.

\section{Our contributions}

In this paper, we resolve the question of whether bin packing with a fixed number of item 
types can be solved in polynomial time.
\begin{theorem} \label{thm:MainTheorem}
For any bin packing instance $(s,a)$ with $s \in [0,1]^d \cap \setQ^d$ and $a \in \setZ_{\geq0}^d$, an optimum 
integral solution can be computed in time $(\log \Delta)^{2^{O(d)}}$ 
where $\Delta := \max\{ \|a\|_{\infty},\|\beta\|_{\infty},4 \}$ where $s_i = \frac{\alpha_i}{\beta_i}$ with $\alpha_i \in \setZ_{\geq 0},\beta_i \in \setZ_{\geq 1}$. % is the largest integer appearing in a denominator of $s_i$ or in a multiplicity $a_i$. %the input. 
%with $M := \max_i\{ \log (b_i), \log (\frac{1}{s_i}) \}$. %  $M^{d^{O(1)}2^{2d}}$
\end{theorem}
This answers an open question posed by McCormick, Smallwood and Spieksma~\cite{BinPackingD2-McCormickSmallwoodSpieksmaSODA97} as well as
by Eisenbrand and Shmonin~\cite{IntegerConesEisenbrandShmoninORL06}. In fact, the first paper
even conjectured this problem to be $\mathbf{NP}$-hard for $d=3$. Moreover the polynomial
solvability for general $d$ was called a ''\emph{hard open problem}'' by Filippi~\cite{BinPackingWithFixedNumberOfWeights-Filippi07}. 

We derive Theorem~\ref{thm:MainTheorem} via the following general theorem for finding conic integer 
combinations in fixed dimension.
\begin{theorem} \label{thm:MainGeneralTheorem}
Given rational polyhedra $P, Q \subseteq \setR^d$ where $P$ is bounded, one can find a vector $y \in  \textrm{int.cone}(P \cap \setZ^d) \cap Q$ and
a vector $\lambda \in \setZ_{\geq0}^{P \cap \setZ^d}$ such that $y = \sum_{x \in P \cap \setZ^d} \lambda_x x$ in time
$\textrm{enc}(P)^{2^{O(d)}} \cdot \textrm{enc}(Q)^{O(1)}$, or decide that no such $y$ exists.
Moreover, the support of $\lambda$ is always bounded by $2^{2d+1}$. %\marginpar{If Eisenbrand-Shmonin is polynomial, then we can put $2^d$ here.}
\end{theorem}
In fact, by choosing $P = \{ {x \choose 1} \in \setR^{d+1}_{\geq 0} \mid s^Tx \leq 1 \}$ and $Q = \{ a\} \times [0,b]$,
we can decide in polynomial time, whether $b$ bins suffice.
Theorem~\ref{thm:MainTheorem} then follows using binary search.
While for the proof strategy it will be convinient that $P$ is bounded, it turns out that one can reduce
the case where $P$ is an unbounded rational polyhedron to Theorem~\ref{thm:MainGeneralTheorem}. However, we postpone that reduction
to Section~\ref{sec:IntConicComForUnboundedPolyhedra}. 

%We assume that $P$ is a bounded polytope; this restriction is not necessary for $Q$.
%In the language of bin packing, 
Our main insight to prove Theorem~\ref{thm:MainGeneralTheorem} lies in the following 
structure theorem which says that, for fixed $d$, 
there is a pre-computable polynomial size set $X \subseteq P \cap \setZ^d$ of special vectors
that are \emph{independent} of the target polytope $Q$ % the multiplicity vector $Q$ 
with the property that, for any $y \in \textrm{int.cone}(P \cap \setZ^d) \cap Q$, there is always a 
conic integer combination that has all
but a constant amount of weight on a constant number of vectors in $X$.

\begin{theorem}[Structure Theorem] \label{thm:StructureTheorem}
  Let $P = \{ x \in \setR^d \mid Ax \leq b\}$ be a polytope with $A \in \setZ^{m \times d}, b \in \setZ^m$ and
  set $\Delta := \max\{ \|A\|_{\infty}, \|b\|_{\infty}, 2\}$. %such that all 
%coefficients are bounded by $\Delta$ in absolute value. % with $\Delta := \max\{ \|A\|_{\infty},\|b\|_{\infty}\} \geq 3$.
Then there exists a set $X \subseteq P \cap \setZ^d$ of size $|X| \leq N := m^d d^{O(d)} (\log \Delta)^{d}$ that can 
be computed in time $N^{O(1)}$  % $m^{O(d)} \cdot d^{O(d^2)} \cdot (\log \Delta)^{O(d^2)}$
with the following property: For any vector $a \in \textrm{int.cone}( P \cap \setZ^d )$ there exists an 
integral vector $\lambda \in \setZ_{\geq0}^{P \cap \setZ^d}$ such that $\sum_{x \in P \cap \setZ^d} \lambda_x x = a$ and
\begin{enumerate}
\item[(1)] $\lambda_x \in \{ 0,1\} \; \;\forall x \in (P \cap \setZ^d)\backslash X$ 
\item[(2)] $|\supp(\lambda) \cap X| \leq 2^{2d}$  
\item[(3)] $|\supp(\lambda) \backslash X| \leq 2^{2d}$.
\end{enumerate}
\end{theorem}
With this structure theorem one can obtain Theorem~\ref{thm:MainGeneralTheorem} simply by 
computing $X$, guessing $\supp(\lambda) \cap X$ and finding the corresponding values
of $\lambda$ and the vectors in $\supp(\lambda) \backslash X$ with an integer program with a constant number of variables.

Bin packing can also be considered as a scheduling problem where the processing times
correspond to the item sizes and the number of machines should be minimized, given
a bound on the makespan. A variety of scheduling
problems in the so-called high multiplicity setting can also be tackled using Theorem~\ref{thm:MainGeneralTheorem}. 
Some of these scheduling applications are described in Section~\ref{sec:HighMultScheduling}. 
For example we can solve in polynomial time the high multiplicity variant of minimizing the
makespan for unrelated machines  with machine-dependent release dates for a fixed number of
job types and machine types. For an overview over the vast literature in high multiplicity scheduling
we refer to the article
of McCormick, Smallwood and Spieksma~\cite{DBLP:journals/mor/McCormickSS01} as well as the one by
Hochbaum and Shamir~\cite{HighMultiplicityScheduling-HochbaumShamir1991}.

%$Q \mid p_{ij},r_{ij},M_J \mid C_{\max}$
%\marginpar{Add ref to Posner paper}
%Theorem~\ref{thm:MainGeneralTheorem}
%immediately implies Theorem~\ref{thm:MainGeneralTheorem}.

%\begin{theorem}[Structure Theorem] \label{thm:StructureTheorem}
%For any Knapsack define by weights $(s_1,\ldots,s_d)$, there exists set 
%of special patterns $S \subseteq \{ x \in \setZ_{\geq0}^d \mid s^Tx \leq 1\}$ of size $O(\log \frac{2}{s_{\min}})^d$
%that can be computed in time $O(\log \frac{2}{s_{\min}})^d$, $s_{\min} := \min_{i=1,\ldots,n}\{ s_i \}$
%such that: for every multiplicity vector $b$, there exists an optimum integral solution $\lambda^*$
%with
%\[
%  |\supp(\lambda^*)| \leq 2^{2d+1} \textrm{  and  }  \sum_{x \in \pazocal{P} \backslash S} \lambda^*_x \leq 2^{2d}.
%\]
%\end{theorem}
%This structure theorem allows us to express bin packing as integer linear program with
%just a constant number of variables which then can be solved via Lenstra's algorithm.

\section{Preliminaries}

In this section we are going to review some known tools that we are 
going to use in our algorithm. The first one is Lenstra's well known
algorithm for integer progamming, that runs in polynomial time as long 
as $d$ is fixed\footnote{Here, the original dependence of \cite{ILPinFixedDimLenstra1981} was $2^{O(d^3)}$ 
which was then improved by Kannan~\cite{IntegerProgramming-Kannan-MOR1987}
to $d^{O(d)}$.}.
\begin{theorem}[Lenstra~\cite{ILPinFixedDimLenstra1981}, Kannan~\cite{IntegerProgramming-Kannan-MOR1987}\label{thm:ILPinFixedDim}]
Given $A \in \setZ^{m \times d}$ and $b \in \setZ^m$ with $\Delta := \max\{ \|A\|_{\infty}, \|b\|_{\infty},2 \}$. Then one
can find an $x \in \setZ^d$ with  $Ax \leq b$ (or decide that none exists) 
in time $d^{O(d)} \cdot m^{O(1)} \cdot (\log \Delta)^{O(1)}$.
\end{theorem}

For a polytope $P \subseteq \setR^d$, the \emph{integral hull} is the convex hull of the
integral points, abbreviated with $P_I := \conv{P \cap \setZ^d}$ and the \emph{extreme points} of $P$ (also called \emph{vertices}) are
denoted by $\vertices(P)$.
 If we consider a low dimensional polytope $P$, then $P$ can indeed contain
 an exponential number of integral points --- but only few of those can be
 extreme points of $P_I$. 
 \begin{theorem}[Cook et al.~\cite{IntegerPointsInPolyhedra-CookHartmannKannanMcDiarmid-Combinatorica92, ComplexityOfIntegerHull-TechReport-Hartmann1988}] \label{thm:NumberOfExtremePoints} % Cook, Hartmann, Kannan, McDiarmid
 Consider any polytope $P = \{ x \in \setR^d \mid Ax \leq b \}$ with  $A \in \setZ^{m \times d}$, $b \in \setZ^m$  and $\Delta := \max\{ \|A\|_{\infty},\|b\|_{\infty},2\}$. 
 Then $P_I = \conv{ P \cap \setZ^d }$ has at most $m^d \cdot (O(\log \Delta))^{d}$ many extreme points.
 In fact a list of extreme points can be computed in time $d^{O(d)} (m\cdot \log(\Delta))^{O(d)}$.
\end{theorem}
This bound is essentially tight. B\'ar\'any, Howe and Lov\'asz~\cite{LowerBoundVerticesOfPI-BaranyHoweLovasz1992} found a family of polytopes $P \subseteq \setR^d$ (simplices, in fact)
such that $P_I$ has $\Omega(\varphi^{d-1})$ many extreme points where $\varphi$ is the encoding length of $P$. 
% % See \footnote{See \url{http://ecommons.library.cornell.edu/handle/1813/8702}}
A simple fact that we use frequently throughout the paper is the following: 
\begin{lemma} \label{lem:InfinityNormOfVertex}
  Let $P = \{ x \in \setR^d \mid Ax \leq b\}$ be a rational polyhedron where $A \in \setZ^{m \times d}$ and $b \in \setZ^m$ and set
  $\Delta := \max\{ \|A\|_{\infty},\|b\|_{\infty}\}$. Then any $x \in \vertices(P)$ has $\|x\|_{\infty} \leq M$ where $M := d! \cdot \Delta^d$. Moreover, if $P$ is bounded then  $P \subseteq [-M,M]^d$.
\end{lemma}
\begin{proof}  
  By Cramer's rule, the coordinates of a vertex $x$ of $P$ are of the form $\det(Q)/\det(R)$
    where $Q$ and $R$ are $d \times d$ matrices filled with entries from $\{ -\Delta,\ldots,+\Delta\}$.
    Then $\|x\|_{\infty} \leq |\det(Q)| \leq d! \Delta^d = M$ as one can see from Laplace formula.
      The ``moreover'' part follows from the observation that the $\| \cdot \|_{\infty}$-norm is maximized at a vertex
      if $P$ is bounded.
\end{proof}

We will later refer to the coefficients $\lambda_x$ as the \emph{weight} given to $x$.
For a vector $a \in \cone(X)$ we know by \emph{Carath{\'e}odory's Theorem} that there is always a 
corresponding vector $\lambda \geq \bm{0}$ with at most $d$ non-zero entries and $a = \sum_{x \in X} \lambda_x x$.
One may wonder how many points $x$ are
actually needed to generate some point in the integer cone.  
%While the integral case is definitely more complex, 
In fact, at least under the additional assumption that $X$ is the set of integral
points in a convex set, %is closed under convex combinations, 
one can show that $2^d$ points suffice\footnote{For arbitrary $X \subseteq \setZ^d$, one can show that a support of at most $O(d \log (dM))$ suffices, where $M$ is the largest coefficient in a vector in $X$~\cite{IntegerConesEisenbrandShmoninORL06}.}.
The arguments are crucial for our proofs, so we replicate the proof
of \cite{IntegerConesEisenbrandShmoninORL06} to be selfcontained.
%the set of integral points a convex set, then $b$ can always be 
%combined using only $2^d$ points. 

\begin{lemma}[Eisenbrand and Shmonin~\cite{IntegerConesEisenbrandShmoninORL06}] \label{lem:ExistenceOfIntegralSolWithSupport2d}
For any polytope $P \subseteq \setR^d$ and any integral vector $\lambda \in \setZ_{\geq0}^{P \cap \setZ^d}$ there exists 
a  $\mu \in \setZ_{\geq0}^{P \cap \setZ^d}$ such that $|\supp(\mu)| \leq 2^d$ and $\sum_{x} \mu_xx = \sum_{x} \lambda_x x$.
%Moreover $\supp(\mu) \subseteq \textrm{conv}( \supp(\lambda))$.
\end{lemma}
\begin{proof}
For the sake of simplicity we can replace the original $P$ with $P := \conv{ x \mid \lambda_x > 0}$
without changing the claim.
Let $f : \setR^d \to \setR$ be any strictly convex function, i.e. in particular 
we will use that 
\[
f(\tfrac{1}{2}x+\tfrac{1}{2}y) < \frac{1}{2}(f(x) + f(y)).
\] 
For example  $f(x) = \|x\|_2^2$ does the job.
%We abbreviate $Q := \conv\{  x\in P \cap \setZ^d \mid \lambda_x > 0\}$ as the convex hull of patterns used by $\lambda$.
Let $(\mu_x)_{x \in P \cap \setZ^d}$ be an integral vector with  $\sum_{x \in P \cap \setZ^d} \lambda_x x = \sum_{x \in P \cap \setZ^d} \mu_x x$ that
minimizes the potential function $\sum_{x \in P \cap \setZ^d} \mu_x \cdot f(x)$ (note that there is at least one such solution, namely $\lambda$). 
In other words, we somewhat prefer points
that are more in the ``center'' of the polytope. We claim that indeed $|\supp(\mu)| \leq 2^d$.

For the sake of contradiction suppose that $|\supp(\mu)| > 2^d$. Then there must be two points
$x,y$ with $\mu_x>0$ and $\mu_y>0$ that have the same \emph{parity}, meaning that $x_i \equiv y_i \mod 2$ for all $i=1,\ldots,d$.
Then $z := \frac{1}{2}(x + y)$ is an integral vector and $z \in P$. Now we remove one unit of weight
from both $x$ and $y$ and add 2 units to $z$. 
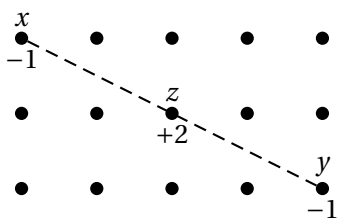
\begin{center}
\begin{pspicture}(0,-0.2)(4,2.2)
\def\mydotsize{2.5pt}
\cnode*(0,0){\mydotsize}{x00} % x
\cnode*(1,0){\mydotsize}{x10}
\cnode*(2,0){\mydotsize}{x20}
\cnode*(3,0){\mydotsize}{x30}
\cnode*(4,0){\mydotsize}{x40}
\cnode*(0,1){\mydotsize}{x01}
\cnode*(1,1){\mydotsize}{x11}
\cnode*(2,1){\mydotsize}{x21} % z
\cnode*(3,1){\mydotsize}{x31}
\cnode*(4,1){\mydotsize}{x41}
\cnode*(0,2){\mydotsize}{x02}
\cnode*(1,2){\mydotsize}{x12}
\cnode*(2,2){\mydotsize}{x22}
\cnode*(3,2){\mydotsize}{x32}
\cnode*(4,2){\mydotsize}{x42} 
\nput[labelsep=2pt]{90}{x02}{$x$}
\nput[labelsep=2pt]{90}{x21}{$z$} %=\frac{1}{2}(x+y)$}
\nput[labelsep=2pt]{90}{x40}{$y$}
\ncline[linestyle=dashed]{x02}{x40}
\nput[labelsep=2pt]{-90}{x02}{$-1$}
\nput[labelsep=2pt]{-90}{x21}{$+2$}
\nput[labelsep=2pt]{-90}{x40}{$-1$}

\end{pspicture}
\end{center}
This gives us another feasible vector $\mu'$. % with $\bm{1}^T\mu' = \bm{1}^T\mu$.
But the change in the potential function is $+2f(z) - f(x) - f(y) < 0$ by strict convexity of $f$, 
contradicting the minimality of $\mu$. 
\end{proof}

In fact, the bound is tight up to a constant factor. As it seems that this has not been observed in 
the literature before, we describe a construction in Section~\ref{sec:EisenbrandShmoninIsTight} 
where a support of size $2^{d-1}$ is actually needed. 

A family of versatile and well-behaved polytopes is those of  \emph{parallelepipeds}.
Recall that
\[
  \parallelepiped = \left\{ v_0 + \sum_{i=1}^k \mu_iv_i \mid -1 \leq \mu_i \leq 1 \; \forall i=1,\ldots,k \right\}
\]
is a \emph{parallelepiped} with \emph{center} $v_0 \in \setR^d$ and \emph{directions}
$v_1,\ldots,v_{k} \in \setR^d$. Usually one requires that the directions are linearly
independent, that means $k \leq d$ and $\parallelepiped$ is $k$-dimensional.
We say that the parallelepiped is \emph{integral} if all its $2^{k}$ 
many vertices are integral. Here is an example of an integral parallelepiped with $d=2$ and $k=2$: 
%\begin{figure}
  \begin{center}
  \begin{pspicture}(0,-0.2)(3,2.2)
    \pspolygon[fillstyle=solid,fillcolor=lightgray](0,0)(2,0)(3,2)(1,2)
    \multido{\N=0+1}{4}{%
      \multido{\n=0+1}{3}{%
      \cnode*(\N,\n){2.5pt}{v\N\n}%
    }}
  \cnode[fillcolor=white,fillstyle=solid](1.5,1){3pt}{v0} \nput[labelsep=2pt]{-90}{v0}{$v_0$}
  \pnode(2.5,1){v1}
  \pnode(2,2){v2}
  \ncline[linewidth=1.0pt,arrowsize=5pt]{->}{v0}{v1} \naput[labelsep=2pt,npos=0.7]{$v_1$}
  \ncline[linewidth=1.0pt,arrowsize=5pt]{->}{v0}{v2} \naput[labelsep=1pt]{$v_2$}
  \rput[c](0.7,0.5){$\pi$}
 %   \psdots(0,0)(1,0)(2,0)(3,0)(0,1)(1,1)(2,1)(3,1)
  \end{pspicture} 
 \end{center}
Note that it is not necessary that all vectors $v_0,\ldots,v_k$ are integral.

\section{Proof of the structure theorem}

In this section we are going to prove the structure theorem. The proof outline
is as follows: we can show that the integral points in a polytope $P$ can be covered
with polynomially many integral parallelepipeds.
The choice for $X$ is then simply the set of vertices of those parallelepipeds.
Now consider any vector $a$ which is a conic integer combination of points in $P$.
Then by Lemma~\ref{lem:ExistenceOfIntegralSolWithSupport2d} we can assume that
$a$ is combined by using only a constant number of points in $P \cap \setZ^d$. 
Consider such a point $x^*$ and say it is used $\lambda^*$ times. We will show that
the weight $\lambda^*$ can be almost entirely redistributed to the vertices of one of the
parallelepipeds containing $x^*$.

Let us make these arguments more formal. 
We begin by showing that all the integer points in a polytope $P$ can indeed be covered with polynomially many
integral parallelepipeds as visualized in Figure~\ref{fig:CoveringWithIntParallelepipeds}.
\begin{figure}
\begin{center}
\psset{unit=0.9cm}
\begin{pspicture}(-0.5,-0.2)(10,6)
\psaxes[ticks=none,labels=none,arrowsize=5pt]{->}(0,0)(0,0)(8.5,5.5)
%\psaxes[arrowsize=5pt]{->}(0,0)(0,0)(8.5,5.5)
\def\sI{0.13}
\def\sII{0.205}
\FPeval{\sIrec}{1.0 / \sI}
\FPtrunc{\sIrec}{\sIrec}{4}
\FPeval{\sIIrec}{1.0 / \sII}
\FPtrunc{\sIIrec}{\sIIrec}{4}
\psline[linewidth=1.5pt](\sIrec,0)(0,\sIIrec)(0,0)
\pspolygon[fillstyle=solid,fillcolor=lightgray](\sIrec,0)(0,\sIIrec)(0,0)
%\pspolygon[fillstyle=solid,fillcolor=lightgray](0,0)(0,4)(1,4)(6,1)(7,0)
\pspolygon[fillstyle=solid,fillcolor=red!50!white,linestyle=dashed,linewidth=0.5pt](0,0)(0,1)(6,1)(6,0)
\pspolygon[fillstyle=solid,fillcolor=green!30!gray,linestyle=dashed,linewidth=0.5pt](0,2)(0,4)(1,4)(1,2)
\psline[linewidth=3pt,linecolor=darkgray](6,1)(7,0)
\pspolygon[fillstyle=solid,fillcolor=blue!50!black,opacity=0.5,linestyle=dashed,linewidth=0.5pt](1,2)(2,3)(6,1)(5,0)
\FPtrunc{\maxI}{\sIrec}{0}
\FPtrunc{\maxII}{\sIIrec}{0}
\FPeval{\maxIp}{\maxI + 1} \FPtrunc{\maxIp}{\maxIp}{0}
\FPeval{\maxIIp}{\maxII + 1} \FPtrunc{\maxIIp}{\maxIIp}{0}
\multido{\I=0+1}{\maxIp}{ % x-Koordinate
 \FPeval{\iter}{((1 - \I * \sI) / \sII) + 1}  % maximales ganzzahliges y+1 so dass (x,y) im Knapsack Polytop fuer x=\I
 \FPtrunc{\iter}{\iter}{0}
 \multido{\N=0+1}{\iter}{ % y-Koordinate
   \psdot(\I,\N)
 }
}
\rput[c](8.5,5pt){$x_1$}
\rput[r](-5pt,5.3){$x_2$}
%\rput[c](4,3.2){$s^Tx \leq 1$}
\end{pspicture}
\caption{Covering the integer points of a polytope with integral parallelepipeds.\label{fig:CoveringWithIntParallelepipeds}}
\end{center}
\end{figure}
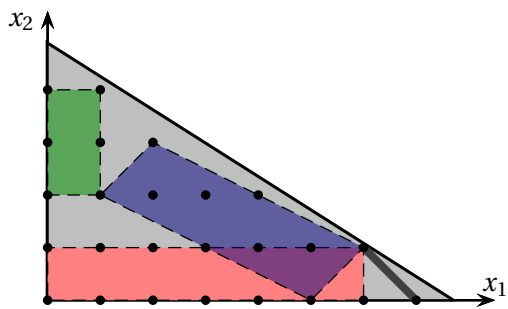
We will say that a polytope $S \subseteq \setR^d$ is \emph{symmetric around center $x_0$}
if $x_0 - x \in S \Leftrightarrow x_0 + x \in S$ for any $x \in \setR^d$.

\begin{lemma} \label{lem:CoveringWithParallelepipeds}
Let $P = \{ x \in \setR^d \mid Ax \leq b\}$ be a polytope where $A \in \setZ^{m \times d}$ and $b \in \setZ^m$ and set $\Delta := \max\{ \|A\|_{\infty},\|b\|_{\infty},2\}$. %described by $m$ inequalities with integral coefficients of absolute value at most $\Delta \geq 2$. % $\Delta := \max\{ \|A\|_{\infty},\|b\|_{\infty}\}$.
Then there exists a set $\parallelepipedset$ of  $|\parallelepipedset| \leq N := m^d d^{O(d)} (\log \Delta)^{d}$ many integral 
parallelepipeds such that 
\[
P \cap \setZ^d \subseteq \bigcup_{\parallelepiped \in \parallelepipedset} \parallelepiped \subseteq P.
\]
Moreover the set $\parallelepipedset$ can be computed in time $N^{O(1)}$. % $m^{O(d)} \cdot d^{O(d^2)} \cdot (O(\log \Delta))^{O(d^2)}$.
\end{lemma}

\begin{figure}
  \begin{center}
\psset{unit=0.5cm}
\begin{pspicture}(-0.2,-1.3)(9,8.7)
\pspolygon[fillstyle=solid,fillcolor=black!20!white,linewidth=1pt](0,0)(10,4)(0,8)%
% lines and red cell C
%
\psline[linewidth=0.75pt,linestyle=dashed](2.7,1.2)(2.7,6.8) \psline[linewidth=0.75pt,linestyle=dashed](5.3,2.2)(5.3,5.8)% \psline(2.7,0)(2.7,7) \psline(5.3,0)(5.3,7) % vert
\psline[linestyle=dashed,linewidth=0.75pt](0,1.3)(8.4,4.65) \psline[linestyle=dashed,linewidth=0.75pt](0,4)(5,6)% \psline(0,1.3)(10,5.3) \psline(0,4)(10,8) % hor
\psline[linestyle=dashed,linewidth=0.75pt](0,6.7)(8.4,3.35) \psline[linestyle=dashed,linewidth=0.75pt](0,4)(5,2)% \psline(0,6.7)(10,2.7) \psline(0,4)(10,0)
\psline[linewidth=2.5pt](0,0)(0,8) \rput[c](0,8.4){$b_i-A_ix=0$}%
%\psline[linewidth=2pt,linestyle=dashed](5.3,0)(5.3,8)%
%\rput[c](5.3,-0.3){$b_i-A_ix=(1+\frac{1}{d})^{\setZ}$}
%\rput[c](5.3,-0.3){$b_i-A_ix$}%
%\rput[c](5.3,-0.9){$=$}%
%\rput[c](5.3,-1.5){$(1+\frac{1}{d})^{\setZ}$}%
%
%
\psline[linestyle=dotted](5.3,2)(5.3,0)
\psline[linestyle=dotted](2.7,1)(2.7,0)
\psbrace[rot=90,braceWidthInner=2pt,braceWidthOuter=2pt,ref=1C,nodesepB=8pt](2.7,0)(5.3,0){$\big\{ x \in \setR^d \mid \alpha_{j_i} \leq A_ix-b \leq \alpha_{j_i+1}\big\}$}
\pspolygon[linestyle=solid,fillstyle=solid,linecolor=red,fillcolor=red!50!lightgray](2.7,2.92)(2.7,5.08)(3.4,5.35)(5.3,4.6)(5.3,3.4)(3.4,2.65)%
%
% ------------
% int hull C_I
\pnode(5.5,6.3){C1} \pnode(4.5,4.5){C2}%
\pspolygon[fillstyle=vlines*,fillcolor=green!40!gray,linecolor=green!50!black,hatchcolor=green!30!black](3,3)(3,5)(4,5)(5,4)(4,3)%
\psdots[linecolor=black,linewidth=1.25pt](3,3)(3,5)(4,5)(5,4)(4,3)(3,4)(4,4)%
\ncline[linecolor=green!50!black,arrowsize=6pt]{->}{C1}{C2} \nput[labelsep=2pt]{90}{C1}{\textcolor{green!50!black}{$C_I$}}%
%
% ------------
%\psdots(0,0)(0,1)(0,2)(0,3)(0,4)(0,5)(0,6)(0,7)(0,8)(1,1)(1,2)(1,3)(1,4)(1,5)(1,6)(1,7)(2,1)(2,2)(2,3)(2,4)(2,5)(2,6)(2,7)(3,2)(3,3)(3,4)(3,5)(3,6)(4,2)(4,3)(4,4)(4,5)(4,6)(5,2)(5,3)(5,4)(5,5)(5,6)(6,3)(6,4)(6,5)(7,3)(7,4)(7,5)(8,4)(9,4)
%\pspolygon[fillstyle=solid,fillcolor=black!40!white](2.7,3)(2.7,5)(3.4,5.2)
%\pnode(0,-8pt){A1} \pnode(2.7,-8pt){A2} \ncline[arrowsize=5pt]{|<->|}{A1}{A2} \nbput[labelsep=2pt]{$(1+\frac{1}{d})^{\setZ}$}
%\pnode(0,-30pt){B1} \pnode(5.2,-30pt){B2} \ncline[arrowsize=5pt]{|<->|}{B1}{B2} \nbput[labelsep=2pt]{$(1+\frac{1}{d})^{j+1}$}
\rput[c](9,5){$P$}%
\pnode(3.4,7){B1} \pnode(3.4,5.4){B2} \ncline[linecolor=red,arrowsize=6pt]{->}{B1}{B2} \nput{90}{B1}{\red{cell $C$}}% 3-7
\end{pspicture}
\caption{Visualization of the slicing of $P$ into cells. \label{fig:SlicingPintoCells}}
\end{center}
\end{figure}

\begin{proof}
  First of all, by Lemma~\ref{lem:InfinityNormOfVertex} every point $x \in P$ has %\footnote{One can use the following argument: the $\| \cdot \|_{\infty}$-norm is maximized by a vertex of $P$. By Cramer's rule, the coordinates of a vertex of $P$ are of the form $\det(Q)/\det(R)$    where $Q$ and $R$ are $d \times d$ matrices filled with entries from $\{ -\Delta,\ldots,+\Delta\}$. Then $|\det(Q)| \leq d! \Delta^d$ as one can see from Laplace formula.}
  $\|x\|_{\infty} \leq d! \cdot \Delta^d$  and hence $|A_ix-b_i| \leq (d+1)\Delta\cdot d!\cdot\Delta^d \leq (d+1)!\cdot\Delta^{d+1}$.
We want to partition the interval $[0,(d+1)! \cdot \Delta^{d+1}]$ into 
smaller intervals $[\alpha_j,\alpha_{j+1}]$ 
such that for any integer values $p,q \in [\alpha_j,\alpha_{j+1}] \cap \setZ$ one has $\frac{p}{q} \leq 1+\frac{1}{d^{2}}$.
For this we can choose $\alpha_j := (1+\frac{1}{d^2})^{j-2}$ for $j=1,\ldots,K$ and $\alpha_0 := 0$. 
The number of intervals is $K \leq O(\log_{1+1/d^2} ((d+1)! \cdot \Delta^{d+1})) \leq O(d^3(\log \Delta + \log d))$.

Our next step is to partition $P$ into \emph{cells} such that points in 
the same cell have roughly the same slacks for all the constraints.
For each sequence  $j_1,\ldots,j_{m} \in \{ 0,\ldots,K-1\}$ we define a cell $C = C(j_1,\ldots,j_{m})$ as
\[
 \left\{ x \in \setR^d \mid \alpha_{j_i} \leq b_i - A_ix \leq \alpha_{j_i+1} \; \forall i\in[m] \right\},
\]
see Figure~\ref{fig:SlicingPintoCells}.
In other words, we partition the polytope $P$ using at most $M := m\cdot K$ many hyperplanes.
By \emph{Buck's Formula}~\cite{BucksFormulaHyperplaneArrangement1943}, the number of $k$-dimensional cells
in a $d$-dimensional arrangement of $M$ hyperplanes is upper bounded by
$f_k^{(d)}(M) := \sum_{i=d-k}^{d}{M \choose i}{i \choose d-k} \leq d \cdot (2M)^d$.
Then the total number of cells of any dimension is generously bounded
by $(d+1) \cdot d \cdot (2M)^d \leq m^d d^{O(d)} (\log \Delta)^d$.
%By a perturbation argument our number of non-empty cells is bounded by the number of
%full dimensional cells in a hyperplane arrangement with $M$ hyperplanes.
%It is a well-known result that the latter quantity is at most ${M \choose 0} + \ldots + {M \choose d} \leq m^d d^{O(d)} (\log \Delta)^d$, 
%see e.g. Matousek~\cite{DiscreteGeometryMatousek2002}. 

Before we continue, we want to comment on the geometry of the cells.
Fix one of those cells
 $C = \big\{ x \in \setR^d \mid \alpha_{j_i} \leq b_i - A_ix \leq \alpha_{j_i+1} \; \forall i\in[m] \big\} \subseteq P$.
 Then for any coordinate $i$ with $0 \leq \alpha_{j_i}<d^2$, we have $|\alpha_{j_i+1}-\alpha_{j_i}| < 1$. That means that the integer hull $\textrm{conv}(C \cap \setZ^d)$ is \emph{flat} in direction $A_i$ unless the cell has a slack of at least $d^2$ in that direction.

Now we proceed to prove that there are only $d^{O(d)}$ integral parallelepipeds necessary to cover the integer points of this cell.
We assume that $C \cap \setZ^d \neq \emptyset$, otherwise there is nothing to do.
Next, fix any integral point $x_0 \in C \cap \setZ^d$ and define a slightly larger polytope 
$Q := \conv{ x_0 \pm (x_0-x) \mid x \in C_I }$, see Figure~\ref{fig:CoveringACellWithFewParallelepipeds}.
By construction,  $Q$ is a symmetric polytope with center $x_0$. Moreover
all its vertices are integral because vertices of $C_I$ and $x_0$ are integral.
%with integral vertices
%containing $C_I$ such that also the center $x_0$ is integral.
The reason why we consider a symmetric polytope
is the following classical theorem which is paraphrased from John: 
\begin{theorem}[John~\cite{EllipsoidsJohn1948}\label{thm:ParaphrasedJohnsTheorem}]
For any polytope $\tilde{P} \subseteq \setR^d$ that is symmetric around center $x_0$,
there are $k \leq \frac{1}{2}d(d+3)$ %\marginpar{T: OK, $d(d+1)/2$ was not enough. Consider the cube in $d=2$. It needs all 4 vertices, but $2(2+1)/2=3$. But for each contact point $x$, we take $x$ and $-x$. . .} 
many extreme points $x_1,\ldots,x_k \in \textrm{vert}(\tilde{P})$ such that $\tilde{P} \subseteq \conv{ x_0 \pm \sqrt{d}\cdot (x_0-x_j) \mid j\in[k]}$.
\end{theorem}
The original statement says that there is in fact an origin-centered ellipsoid $E$ %with center $\bm{0}$ 
with $x_0 + \frac{1}{\sqrt{d}}E \subseteq \tilde{P} \subseteq x_0 + E$. But additionally John's Theorem provides a set of
\emph{contact points} in $\partial P \cap \partial E$ whose convex hull already contains the scaled ellipsoid $x_0 + \frac{1}{\sqrt{d}}E$.
Moreover, the number of necessary contact points is at most $\frac{d}{2}(d+3)$, implying the above statement.

So we apply Theorem~\ref{thm:ParaphrasedJohnsTheorem} to $Q$ with center $x_0$  and
obtain a list of points $x_1,\ldots,x_k \in \textrm{vert}(C_I)$ with $k \leq \frac{1}{2}d(d+3)$ such that 
\[
  C_I \subseteq Q \subseteq \conv{ x_0 \pm \lceil\sqrt{d}\rceil \cdot (x_0-x_j) \mid j\in [k]} =: Q'.
\]
Now it is not difficult to cover $C_I$ with parallelepipeds of the form 
\[
  \parallelepiped(J) := \Big\{ x_0 + \sum_{j\in J} \mu_j(x_j-x_0) \mid  |\mu_j| \leq \lceil\sqrt{d}\rceil \;\; \forall j \in J \Big\}
\]
with $J \subseteq [k]$  and  $\{ x_j-x_0 \mid j \in J\}$ linearly independent. %More precisely, we claim that even $Q' \subseteq \bigcup_{J \subseteq [k]: |J| \leq d} \parallelepiped(J)$.
To see this take any point $x \in Q'$. By Carath{\'e}odory's Theorem, $x$ lies already in the 
convex hull of $x_0$ plus at most $d$ affinely independent vertices of $Q'$, thus there is a subset of 
indices $J \subseteq [k]$ of size $|J| \leq d$ and signs $\varepsilon_j \in \{ \pm 1\}$ with $x \in \conv{\{ x_0 \} \cup \{ x_0 + \varepsilon_j\lceil \sqrt{d}\rceil \cdot (x_j-x_0) \mid j \in J\}}$. Then clearly $x \in \parallelepiped(J)$.

Finally it remains to show that all parallelepipeds $\parallelepiped(J)$ are still in $P$. 
%Note that $\parallelepiped(J) \subseteq Q(d^{3/2})$, thus it suffices to argue that any vertex $y$ of $Q(d^{3/2})$ lies in $P$.
%Each such vertex is of the form $y = x_0 � d^{3/2}(x_0-x)$ with $x \in \textrm{vert}(C_I)$. Then we can verify that
Let $x = x_0 + \sum_{j \in J} \mu_j(x_j-x_0)$ with $|\mu_j| \leq \lceil \sqrt{d} \rceil$. We need to verify that $x$ does not violate a constraint $i \in [m]$. First consider the case that $j_i>0$, that means $C$ does
not lie in the very first slice of constraint $i$. In this case we have
\[
   b_i-A_ix 
 \geq \underbrace{b_i-A_ix_0}_{\geq\alpha_{j_i}} - \sum_{j\in J} \underbrace{|\mu_j|}_{\leq\lceil\sqrt{d}\rceil}\cdot\underbrace{|A_ix_j - A_ix_0|}_{\quad\leq\alpha_{j_i+1}-\alpha_{j_i}\leq\frac{\alpha_{j_i}}{d^2}}  \geq 0
\]
where we crucially use that $|\alpha_{j_i+1} - \alpha_{j_i}| = (1+\frac{1}{d^2})\alpha_{j_i} - \alpha_{j_i} = \frac{\alpha_{j_i}}{d^2}$. If indeed $j_i=0$, then
$|\alpha_{j_i+1}-\alpha_{j_i}| = (1+\frac{1}{d^2})^{-1}<1$ and for integer points $x_j$
and $x_0$ one must have $|A_ix_j - A_ix_0| = 0$.
Finally observe that the number of subsets $J$ of size at most $d$ is $(\frac{1}{2}d(d+3))^d = d^{O(d)}$
which then gives the desired bound. 

Now let us argue how to make this constructive in time $N^{O(1)}$. For each cell $C$, we list the vertices 
of the integer hull $C_I$ in time $d^{O(d)}m^{O(d)}(\log \Delta)^{O(d)}$ by Theorem~\ref{thm:NumberOfExtremePoints}.
Computing the minimum volume ellipsoid containing all those vertices is indeed 
 a semidefinite program that can be solved in time polynomial in the encoding length of the vertices of $C_I$. We refer to Chapter 8 of Boyd and Vandenberghe~\cite{Convex-Optimization-BoydVandenberghe-2014} for details.
The contact points can be inferred from the dual solution of this SDP and the associated parallelepipeds
can be easily computed. 
% Every step can be made constructive at the expense of a fixed polynomial 
% blowup (with the exponent independent of $d$), just for optimizing over $C_I$ one needs
%to solve an integer program which accounts for  extra $d^{O(d)}$ factor. This can be absorbed into
% the constants.
\end{proof}

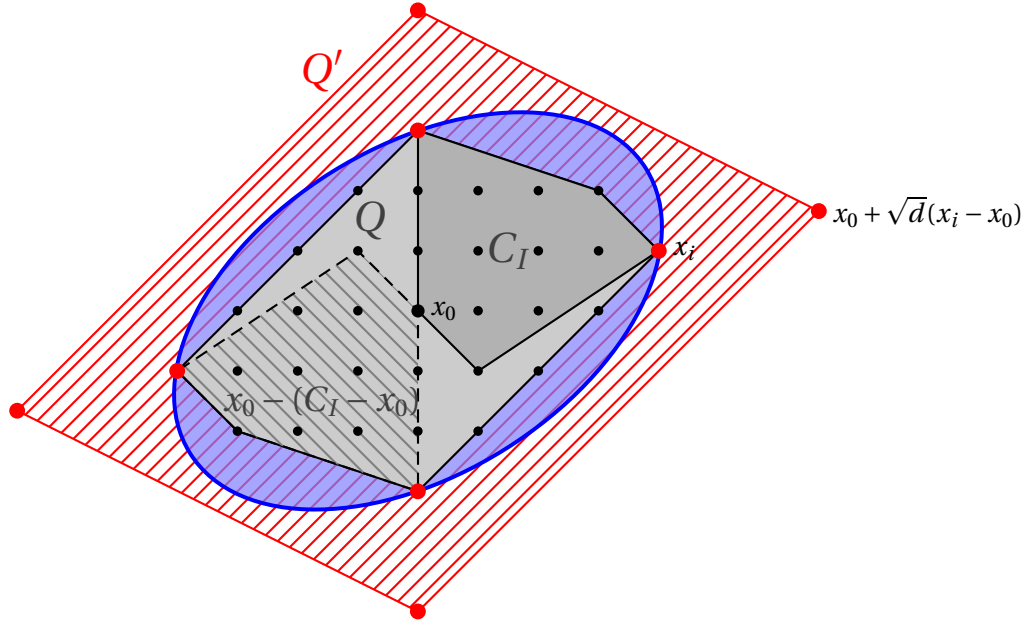
\begin{figure}
\begin{center}
\psset{unit=0.6cm}
\begin{pspicture}(-4,-5)(4,3.5)
\def\mydotsize{2.5pt}
% contact points scaled by 5/3
\pspolygon[fillstyle=hlines,linecolor=red,hatchcolor=red](0,5)(6.66,1.66)(0,-5)(-6.66,-1.66)
\psdots[linecolor=red,linewidth=2pt](0,5)(6.66,1.66)(0,-5)(-6.66,-1.66)
\rput[c](-1.6,4.0){\LARGE{\red{$Q'$}}}
% Ellipse
\rput{32}(0,0){\psellipticwedge[linewidth=1.5pt,linecolor=blue,fillstyle=solid,fillcolor=blue!50!white,opacity=0.7](0,0)(4.5,2.75){0}{-1}} 
% symmetric closure 
\pspolygon[fillstyle=solid,fillcolor=black!20!white](0,3)(3,2)(4,1)(0,-3)(-3,-2)(-4,-1) \rput[c](-0.8,1.5){\LARGE{\darkgray{$Q$}}} 
% The cell C_I
\pspolygon[fillstyle=solid,fillcolor=black!30!white](0,0)(0,3)(3,2)(4,1)(1,-1) \rput[c](1.5,1){\LARGE{\darkgray{$C_I$}}}
\pspolygon[linestyle=dashed,fillstyle=vlines,hatchcolor=gray](0,0)(0,-3)(-3,-2)(-4,-1)(-1,1) \rput[c](-1.6,-1.5){\Large{\darkgray{$x_0-(C_I-x_0)$}}} % <- -C_I
\psdots(0,0)(0,1)(0,2)(0,3)(1,-1)(1,0)(1,1)(1,2)(2,0)(2,1)(2,2)(3,1)(3,2)(4,1)
\psdots(0,0)(0,-1)(0,-2)(0,-3)(-1,1)(-1,0)(-1,-1)(-1,-2)(-2,0)(-2,-1)(-2,-2)(-3,-1)(-3,-2)(-4,-1)(3,0)(2,-1)(1,-2)(-3,0)(-2,1)(-1,2) % points in Q/C_I
\cnode*(0,0){\mydotsize}{x0} \nput[labelsep=2pt]{0}{x0}{$x_0$}
% contact points
\psdots[linecolor=red,linewidth=2pt](0,3)(4,1)(0,-3)(-4,-1)
\pnode(4,1){xi} \pnode(6,-1){xil} %\ncline[arrowsize=6pt,linewidth=1pt,nodesepB=3pt]{->}{xil}{xi} %\nput{-90}{xil}{contact point $x_i$}
\pnode(6.66,1.66){sxi}
\nput[labelsep=5pt]{0}{xi}{$x_i$}
\nput[labelsep=5pt]{90}{sxi}{$x_0+\sqrt{d}(x_i-x_0)$}
% scaled *2contact points
%\pspolygon(0,6)(8,2)(0,-6)(-8,-2)
%\psdots[linecolor=red,linewidth=2pt](0,6)(8,2)(0,-6)(-8,-2)

\end{pspicture}
\caption{Visualization of covering the integer points in a cell $C_I$: Start by obtaining the symmetric closure $Q$. Then compute the contact points of a minimum volume ellipsoid containing $Q$. Scale those points with $\sqrt{d}$ to obtain a polytope $Q'$ with only $O(d^2)$ vertices containing $C_I$. Then extend a triangulation of $Q'$ to $d^{O(d)}$ many parallelepipeds.\label{fig:CoveringACellWithFewParallelepipeds}}
\end{center}
\end{figure}
Note that one could have used the following simpler arguments to obtain a 
weaker bound that still leads to a polynomial time algorithm for bin packing if $d$ is constant: first of all, each cell is defined by selecting $m$ values $\alpha_{j_i} \in \{ 0,\ldots,K-1\}$, hence the total number of cells is trivially
upper bounded by $K^m$. Then every cell $C_I$
has polynomially many vertices, hence it can be partitioned into polynomially many
simplices. Then each simplex can be extended to a parallelepiped, whose union again covers $C_I$.

As a side remark, the partitioning with shifted hyperplanes was used before e.g. 
in~\cite{IntegerPointsInPolyhedra-CookHartmannKannanMcDiarmid-Combinatorica92} to bound the number of extreme points of $\conv{P \cap \setZ^d}$.
The next lemma says why parallelepipeds are so useful. Namely the weight of
any point in it can be almost completely redistributed to its vertices.
\begin{lemma} \label{lem:RedistritionWithinParallelepiped}
Given an integral parallelepiped $\parallelepiped$ % 
with vertices $X := \textrm{vert}(\parallelepiped)$. % whose vertices $X = \textrm{vert}(\parallelepiped)$ are integral. 
Then for any $x^* \in \parallelepiped \cap \setZ^d$ and $\lambda^* \in \setZ_{\geq0}$ there is an 
integral vector $\mu \in \setZ_{\geq0}^{\parallelepiped \cap \setZ^d}$ such that
\begin{enumerate}
\item[(1)] $\lambda^*x^* = \sum_{x \in \parallelepiped \cap \setZ^d} \mu_x x$ 
\item[(2)] $|\supp(\mu) \backslash X| \leq 2^d$
\item[(3)] $\mu_x \in \{ 0,1\} \; \forall x \notin X$.
\end{enumerate}
\end{lemma}
\begin{proof}
Let $\parallelepiped = \{ v_0 + \sum_{i=1}^{k} \alpha_iv_i \mid |\alpha_i| \leq 1 \; \forall i=1,\ldots,k\}$ where $v_0$ is the (not necessarily integral) center of $\parallelepiped$.
Consider a vector $\mu$ that satisfies (1) and minimizes the potential
function $\sum_{x \notin X} \mu_x$ (i.e. the weight
that lies on non-vertices of $\parallelepiped$). We claim that $\mu$ also satisfies (2) and (3).

First consider the case that there is some point $x \in \parallelepiped \cap \setZ^d$
that is not a vertex and has $\mu_x \geq 2$. We write $x = v_0 + \sum_{i=1}^{k} \alpha_i v_i$ with $|\alpha_i| \leq 1$.
Let\footnote{Recall that $\sign(\alpha) = \begin{cases} 1 & \alpha \geq 0 \\ -1 & \alpha < 0 \end{cases}$} 
$y := v_0 + \sum_{i=1}^k \sign(\alpha_i) \cdot v_i \in \parallelepiped \cap \setZ^d$ be the vertex of $\parallelepiped$ that we obtain by rounding $\alpha_i$ to $\pm 1$, see Figure~\ref{fig:RedistributionInParallelepiped}. 
The mirrored point $z = x + (x-y) = v_0 + \sum_{i=1}^{k} (2\alpha_i-\sign(\alpha_i))\cdot v_i$ lies in $\parallelepiped$
as well and is also integral. Here we use the fact that $-1 \leq (2\alpha_i-\sign(\alpha_i)) \leq 1$. As $x = \frac{1}{2}(y+z)$, we can reduce the weight on $x$ by $2$
and add $1$ to $\mu_y$ and $\mu_z$. We obtain again a vector that satisfies (1), but the weight 
$\sum_{x \notin X} \mu_x$ has decreased. 
\begin{figure}
\begin{center}
\psset{xunit=1cm,yunit=1.2cm}
\begin{pspicture}(0,-0.2)(4,2.2)
\pspolygon[linewidth=1pt,fillstyle=solid,fillcolor=lightgray](0,0)(3,0)(4,2)(1,2)
\psline[linestyle=dashed](1.5,0)(2.5,2)
\psline[linestyle=dashed](0.5,1.0)(3.5,1)
\rput[c](3,1.5){\Huge{$\gray{\parallelepiped}$}}
\def\mydotsize{2.5pt}%
\cnode*(0,0){\mydotsize}{x00}
\cnode*(3,0){\mydotsize}{x10}
\cnode*(1,2){\mydotsize}{x01} % y
\cnode*(4,2){\mydotsize}{x11}
\cnode*[linecolor=gray](1.5,1.25){\mydotsize}{x} % x
\cnode*[linecolor=gray](2,0.5){\mydotsize}{z}
%\cnode*[linecolor=gray](3.5,0.5){\mydotsize}{y} \nput[labelsep=2pt]{-30}{y}{$y$}
%\cnode*[linecolor=gray](4,1){\mydotsize}{z} \nput[labelsep=2pt]{90}{z}{$z$}
\ncline[arrowsize=5pt,linewidth=1.5pt]{->}{x}{x01} %\nbput[labelsep=0pt,npos=0.4]{$d$}
\ncline[arrowsize=5pt,linewidth=1.5pt]{->}{x}{z} %\naput[labelsep=0pt,npos=0.4]{$-d$}
%\nput{-90}{x10}{$\in X$}
\nput{90}{x01}{$y$}
\nput[labelsep=2pt]{45}{x}{$x$}
\nput[labelsep=2pt]{0}{z}{$z$}
% mu0=0.4, mu1=0.3  mu2=0.1
\end{pspicture}
\caption{Weight of $y$ is redistributed to vertex in parallelepiped. \label{fig:RedistributionInParallelepiped}}
\end{center}
\end{figure}
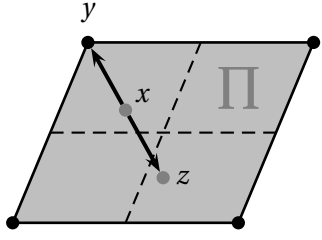

So it remains to see what happens when all vectors in $(\parallelepiped \cap \setZ^d) \setminus X$ carry weight at most 1.
Well, if these are at most $2^d$, then we are done. Otherwise, we can reiterate
the
arguments from Lemma~\ref{lem:ExistenceOfIntegralSolWithSupport2d}. 
There will be 2 points of the same parity, which can be joined
to create a new point carrying weight at least 2 and part of this weight can be redistributed to a vertex. 
This shows the claim. 
\end{proof}

Now we simply combine Lemmas~\ref{lem:ExistenceOfIntegralSolWithSupport2d}, \ref{lem:CoveringWithParallelepipeds}
and \ref{lem:RedistritionWithinParallelepiped}.
\begin{proof}[Proof of Structure Theorem~\ref{thm:StructureTheorem}]
We choose $X$ as the at most $N = m^d d^{O(d)} (\log \Delta)^{d}$ many vertices of parallelepipeds $\parallelepipedset$
that are constructed in Lemma~\ref{lem:CoveringWithParallelepipeds} in running time $N^{O(1)}$
(there is an extra $2^d$ factor, 
that accounts for the maximum number of vertices per parallelepiped; 
this is absorbed by the $O$-notation). Now consider any vector $a \in \textrm{int.cone}(P \cap \setZ^d)$. 
By Lemma~\ref{lem:ExistenceOfIntegralSolWithSupport2d} there is a vector $\mu \in \setZ_{\geq0}^{P \cap \setZ^d}$ 
with $|\supp(\mu)| \leq 2^d$ and $a = \sum_x \mu_x \cdot x$. For every $x$ with $\lambda_x>0$ we consider a
parallelepiped $\parallelepiped \in \parallelepipedset$ with $x \in \parallelepiped \cap \setZ^d$. 
Then we use Lemma~\ref{lem:RedistritionWithinParallelepiped} to redistribute the weight from $x$
to the vertices of $\parallelepiped$. For each parallelepiped, there are at most $2^d$ non-vertices with
a weight of 1. In the case in which a vector is used by several parallelepipeds, we
can further redistribute its weight to the vertices of one of the involved parallelepipeds.
This process terminates as the total weight on $X$ keeps increasing.
%When combining the parallelepipeds
% (except of a weight of at most $2^d$ that may remain). 
We denote the new solution by $\lambda$. As we are using at most $2^d$ parallelepipeds, 
we have %It is not difficult to see that
$|\supp(\lambda) \cap X| \leq 2^d \cdot 2^d$ and $|\supp(\lambda) \backslash X| \leq 2^d \cdot 2^d$. 
%$|\supp(\lambda)| \leq 2^d \cdot (2^d+2^d)$ and $\sum_{x \in (P \cap \setZ^d) \backslash X} \lambda_x \leq 2^d \cdot 2^d$. 
\end{proof}

\section{Proof of the main theorem}

Now that we have the Structure Theorem, the claim of 
Theorem~\ref{thm:MainGeneralTheorem} is easy to show.
\begin{proof}[Proof of Main Theorem~\ref{thm:MainGeneralTheorem}]
  As both polyhedra $P$ and $Q$ are assumed to be rational, we can write
  them as
  $P = \{ x \in \setR^d \mid Ax \leq b\}$ and $Q = \{ x \in \setR^d \mid \tilde{A}x \leq \tilde{b}\}$ with $A,b,\tilde{A},\tilde{b}$ integral
  while the encoding length increases by at most a polynomial factor.
%  be the given polytopes. 
%Here we assume that the coefficients in the inequality description are integral and 
%the numbers in $A,b$ and $\tilde{A},\tilde{b}$ are bounded  in absolute value by $\Delta$ and $\tilde{\Delta}$, respectively.
  We abbreviate $\Delta := \max\{ \|A\|_{\infty},\|b\|_{\infty},2\}$
  and $\tilde{\Delta} := \max\{ \|\tilde{A}\|_{\infty}, \|\tilde{b}\|_{\infty},2 \}$. 

We compute the set $X$ of size at most $N := m^{d} d^{O(d)} (\log \Delta)^{d}$
from Theorem~\ref{thm:StructureTheorem} for the polytope $P$ in time
$N^{O(1)}$. Now let $y^* \in \textrm{int.cone}(P \cap \setZ^d) \cap Q$ be an unkown target vector.
Then we know by Theorem~\ref{thm:StructureTheorem} that there is a vector $\lambda^* \in \setZ^{P \cap \setZ^d}_{\geq 0}$ 
such that $\sum_{x \in P \cap \setZ^d} \lambda_x^* x = y^*$, $|\supp(\lambda^*) \cap X| \leq 2^d$, $|\supp(\lambda^*) \backslash X| \leq 2^d$
and $\lambda_x^* \in \{ 0,1\}$ for $x \in (P \cap \setZ^d) \backslash X$.

After enumerating all subsets of $X' \subseteq X$ of cardinality $|X'| \leq 2^{2d}$
and running the subsequent test for each of them, we may
assume to know the set $X' = X \cap \supp(\lambda^*)$. Note that
there are at most $N^{2^{2d}}$ sets to test\footnote{Actually we know that $X'$ consists of the vertices of at most $2^d$ parallelepipeds, thus it suffices to incorporate a factor of $N^{2^d}$, but the improvement would
  be absorbed by the $O$-notation later, anyway.}. 
At the expense of another factor $2^{2d}+1$ we guess the number $k = \sum_{x \notin X'} \lambda^*_x \in \{ 0,\ldots,2^{2d}\}$ of
extra points. Now we can set up an integer program with few variables. 
We use variables $\lambda_x$ for $x \in X'$ to determine the correct multiplicities of
the points in $X$. Moreover, we have variables $x_1,\ldots,x_k \in \setZ_{\geq0}^{d}$ to determine
which extra points to take with unit weight. Additionally we
use a variable $y \in \setZ^d$ to denote the target vector in polyhedron $Q$.
The ILP is then of the form 
\begin{eqnarray*}
% \sum_{x \in X'} \lambda_x + k &\leq& \delta\\
 Ax_i &\leq& b \quad \forall i=1,\ldots,k \\
 \sum_{x \in X'} \lambda_x x + \sum_{i=1}^k x_i  &=& y \\
 \tilde{A}y &\leq& \tilde{b} \\ 
 \lambda_x &\in& \setZ_{\geq0} \quad \forall x \in X' \\
 x_i,y &\in& \setZ^d \quad \forall i=1,\ldots,k 
\end{eqnarray*}
and given that we made the correct guesses, this system has a solution.
The number of variables is $|X'| + (k+1)d \leq 2^{O(d)}$ and the number of
constraints is $km + d + \tilde{m} + |X'|d  = 2^{O(d)}m+\tilde{m}$ as well. Note that the largest coefficient in the ILP
is at most $\Delta' := \max\{ d! \cdot \Delta^d, \tilde{\Delta}\}$ as $\|x\|_{\infty} \leq d! \cdot \Delta^d$ for $x \in X'$ (see Lemma~\ref{lem:InfinityNormOfVertex}), as well as $\max\{ \|A\|_{\infty},\|b\|_{\infty}\} \leq \Delta$ and $\max\{ \|\tilde{A}\|_{\infty},\|\tilde{b}\|_{\infty}\} \leq \tilde{\Delta}$. Hence the ILP can be solved in time $(2^{O(d)})^{2^{O(d)}} \cdot (2^{O(d)}m+\tilde{m})^{O(1)} \cdot (\log \Delta')^{O(1)} \leq 2^{2^{O(d)}} \cdot \textrm{enc}(P)^{O(1)} \cdot \textrm{enc}(Q)^{O(1)}$
via Theorem~\ref{thm:ILPinFixedDim}.

We note that the largest factor in upper bounding the total running time is indeed
the term $N^{2^{2d}} \leq \enc(P)^{2^{O(d)}}$ required to guess the correct choice of $X'$.
Multiplying the different terms results in a total running time of the form $\enc(P)^{2^{O(d)}} \cdot \enc(Q)^{O(1)}$.  %$(m + \log \Delta)^{2^{O(d)}} \cdot (\tilde{m}+\log \tilde{\Delta})^{O(1)}$. 
%\]
\end{proof}
Note that the structure theorem uses that integer combination are taken w.r.t. a set $X = P \cap \setZ^d$
that is closed under taking convex combinations. On the other hand, without any assumption on the structure of $X$, the test $\textrm{int.cone}(X) \cap Q \neq \emptyset$ is $\mathbf{NP}$-hard even for
$d=1$ and $Q$ being a single point.
To see this, recall that given positive integers $a_1,\ldots,a_n$ with parameters $k \in \setN$ and $S \geq 2$,
it is $\mathbf{NP}$-hard to decide whether exactly $k$ of the numbers can be added to give
exactly $S$~\cite{GareyJohnson79}. Then for $S' := 2(a_1+\ldots+a_n)$, this decision problem is equivalent to $\textrm{int.cone}(\{S'+a_1,\ldots,S'+a_n\}) \cap \{kS'+S\} \neq \emptyset$.

\section{High multiplicity scheduling\label{sec:HighMultScheduling}}

In this section, we want to demonstrate the power and versatility of our
method by describing, how most scheduling problems with a constant number
of job types and a constant number of machine types can be solved in polynomial time.
For didactic purposes, we begin with a simple application and then discuss
how the approach can be generalized to capture many more scheduling problems. 

\subsection{Cutting Stock}

In the \emph{cutting stock} problem, we are given \emph{item sizes}
 $s_{1},\ldots,s_{d} \in \setN$ and multiplicity $a_j \in \setN$ for each item $j \in [d]$. Additionally we have a list of 
 $m$ \emph{bin types}, where bin type $i \in [m]$ has \emph{capacity} $w_i \in \setN$ and \emph{cost} $c_i \in \setN$. The task is to assign all the items to bins so that
 the items assigned to each  bin do not exceed its capacity. The objective function
 is to minimize the cost where we pay an amount of $c_i$ for each bin type $i$ that is being used. The study of this problem goes back at least to the 1960's to the classical paper of Gilmore and Gomory~\cite{Gilmore-Gomory61}. In particular bin packing is equivalent to the case where $m=1$.

 We will now explain how this problem can be solved using the test from Theorem~\ref{thm:MainGeneralTheorem}. First, using binary search we can reduce the optimization variant of cutting stock to the decision variant where for a given parameter $T \in \setZ_{\geq 0}$ we have to decide whether
 there is a solution of objective function at most $T$.
 The intuitive way to model the problem is by defining $m$ \emph{knapsack polytopes} that include a coordinate for the objective function as well as a \emph{target polytope} by setting
 \[
P_i := \Big\{ \begin{pmatrix} x \\ c_i \end{pmatrix} \in \setR^{d+1} \mid \sum_{j=1}^d x_js_j \leq w_i\textrm{ and }x_j \geq 0 \; \forall j \in [d] \Big\} \quad \forall i \in [m] \quad \quad \textrm{and} \quad Q := \Big\{ \begin{pmatrix} a \\ t \end{pmatrix} \in \setR^{d+1} \mid 0 \leq t \leq T \Big\}
 \]
 Then the cutting stock problem indeed has a solution with objective function value at most $T$, if and only if
 \[
\textrm{int.cone}\big( (P_1 \cup \ldots \cup P_m) \cap \setZ^{d+1}\big) \cap Q \neq \emptyset.
 \]
 The issue with this approach is that in general the union of polytopes is non-convex and Theorem~\ref{thm:MainGeneralTheorem} only applies to convex polytopes. But as we have integer variables at our disposal, this issue can be fixed using the \emph{Big-$M$ method}.

\begin{theorem}
The cutting stock problem with $d$ different item types and $m$ different bin types
can be solved in time $(\log \Delta)^{2^{O(d+m)}}$ where $\Delta := \max\{ \|c\|_{\infty},\|w\|_{\infty},\|s\|_{\infty},4\}$. %$M = M := \max\{ s_j,w_i,c_i,T \mid j \in [d],i \in [m]\}+1$. % is the largest number in the input.
\end{theorem}

\begin{proof}
 We  set  $M := \max\{ w_i,c_i \mid i \in [m]\}$ and define
 \[
   \tilde{P} := \left\{ \begin{pmatrix} x \\ y \\ z \end{pmatrix} \in \setR^{d+1+m}_{\geq 0} \mid \begin{array}{rcl}  \sum_{j=1}^d s_jx_j &\leq& w_i + (1-z_i) M \; \forall i \in [m] \\
                                                                                                    y &\leq& c_i + (1-z_i)M \; \forall i \in [m] \\
                                                                                                    y &\geq & c_i - (1-z_i)M \; \forall i \in [m] \\
   \sum_{i \in [m]} z_i &=& 1
                                                                                                  \end{array} \right\}
                                                                                                \quad \textrm{and} \quad \tilde{Q} := \left\{\begin{pmatrix} a \\ t \\ z\end{pmatrix} \mid 0\leq t \leq T   \right\}
 \]
 Here, $z_i$ is the decision variable telling whether the vector $(x,y,z)$ originates from bin type $i$. First, we claim that
  \begin{equation}
\tilde{P} \cap \setZ^{d+1+m} = \bigcup_{i=1}^m \left\{ \begin{pmatrix} x \\ c_i \\ e_i \end{pmatrix} \mid x \in \setZ_{\geq 0}^d \;\; \textrm{and} \;\; \sum_{j=1}^d s_jx_j \leq w_i \right\} \label{eq:UnionForCuttingStock}
  \end{equation}
where $e_i$ is the $i$th unit vector in $\setR^m$.  So, consider a vector $(x,y,z) \in \tilde{P} \cap \setZ^{d+1+m}$. Then by integrality and the constraint $\sum_{i \in [m]} z_i=1$, there is a fixed index $i$ with $z = e_i$. From the other constraints in $\tilde{P}$ we can infer that $y = c_i$ and $\sum_{j=1}^d s_jx_j \leq w_i + (1-z_i)M = w_i$. In reverse, we can see that any vector $(x,c_i,e_i)$ that is contained in the right hand side of \eqref{eq:UnionForCuttingStock} satisfies the constraints of $\tilde{P}$. In particular for $i' \neq i$ one has $\sum_{j=1}^d s_jx_j \leq w_i \leq w_{i'} + (1-z_{i'})M$
  as well as $c_i \leq c_{i'} + (1-z_{i'})M$.
  Hence the test
  \[
\textrm{int.cone}\big(\tilde{P} \cap \setZ^{d+1+m}\big) \cap \tilde{Q} \neq \emptyset
  \]
  correctly decides whether the cutting stock instance has a solution of value at most $T$. It remains to analyze the running time. 
  The maximum absolute value in the constraint matrices describing $\tilde{P}$
  and $\tilde{Q}$ is upper bounded by $\max\{ \Delta,2M,T\} \leq  d\Delta^2$ where we use $M \leq \Delta$ and\footnote{We can limit $T$ by the following argument: the cutting stock instance can only be feasible if $\|s\|_{\infty} \leq \|w\|_{\infty}$. Then assigning each item to a single bin of maximum capacity results in a solution of cost at most $\|a\|_1 \cdot \|c\|_{\infty}$.} $T \leq \|a\|_1\cdot \|c\|_{\infty} \leq d\Delta^2$. Moreover the number of variables is $d+1+m$ and the number of constraints is $O(d+m)$.
  Then the encoding length of both polytopes are bounded by $\textrm{enc}(\tilde{P}),\textrm{enc}(\tilde{Q}) \leq (m+d)^{O(1)}\cdot O(\log(d\Delta))$ and the claim follows from Theorem~\ref{thm:MainGeneralTheorem}.
\end{proof}
From a more abstract point of view, we have solved the cutting stock problem by
reducing it to a test of the form
 \[
\textrm{int.cone}\big( X_1 \cup \ldots \cup X_m \big) \cap Q \neq \emptyset
\]
with $X_i = P_i \cap \setZ^{d}$, where $P_1,\ldots,P_m$ are polytopes and hence
convex sets. While this approach was sufficient for cutting stock, it is 
more restricted than necessary. In fact, we can also handle sets of the form
\[
 X_i = \big\{ x \in \setZ^{d} \mid \exists y \in \setZ^{d_i} : A^ix + B^iy \leq b^i \big\}.
\]
Such sets are called \emph{integer projections} and are extremely powerful.
One should note that integer projections are in general not closed under
taking convex combinations: it is possible that $x,y \in X_i$ with $\frac{1}{2}x + \frac{1}{2}y$ integral while $(\frac{1}{2}x + \frac{1}{2}y) \notin X_i$.
In fact, one can even show that \emph{every} finite set is an integer projection.
%, we cannot imply that 
%The crucial quantity that will determine
Whether an approach using integer projections is efficient 
will be determined largely by the number $d_i$ of extra variables as we will see.

%\begin{proof}
%We can represent the multiset of items  $x \in \setZ_{\geq 0}^d$ that fit into 
%a bin of size $w_j$ by $X_j := \{ x \in \setZ_{\geq 0}^d \mid \sum_{i=1}^d s_ix_i \leq w_j\}$.
%Hence, using $c_j$ as machine cost, the claim follows immediately from 
%Theorem~\ref{thm:GeneralScheduling}.
%\end{proof}

\subsection{Integer conic combinations of unions of integer projections}

Motivated by further applications to high multiplicity scheduling, we now give
a generalization of the test from Theorem~\ref{thm:MainGeneralTheorem} to
unions of integer projections.
\begin{theorem} \label{thm:IntegerConicCombinationsOfUnionsOfIntegerProjections}
  Given rational polytopes $P_1,\ldots,P_m \subseteq \setR^{d+d_i}$ and a rational polyhedron $Q \subseteq \setR^d$,
  define $X_i := \{ x \in \setZ^d \mid \exists y \in \setZ^{d_i} : (x,y) \in P_i \}$ and $d_{\textrm{aux}} := \max\{ d_i : i \in [m]\}$. Then there is an algorithm that decides correctly whether
  \[
\textrm{int.cone}(X_1 \cup \ldots \cup X_m) \cap Q \neq \emptyset
\]
in time $(\sum_{i=1}^m \enc(P_i))^{2^{O(d + d_{\textrm{aux}}+m)}} \cdot \textrm{enc}(Q)^{O(1)}$. In the affirmative case the algorithm provides a vector $\lambda = (\lambda_{i,x})_{i \in [m], x \in X_i}$ with $\lambda_{i,x} \in  \setZ_{\geq 0}$ and $(\sum_{i=1}^m \sum_{x \in X_i} \lambda_{i,x} \cdot x) \in Q$.
Moreover, the size of the support of $\lambda$ is bounded by $2^{2(d+d_{\textrm{aux}}+m)+1}$.
\end{theorem}

\begin{proof}
  Let us write $P_i = \{ (x,y) \in \setR^{d+d_i} \mid A^ix + B^iy \leq b^i \}$ and $Q= \{ x \in \setR^d \mid A'x \leq b' \}$, again assuming after multiplying with the least common denominator that $A^i,B^i,b^i$ have only integer coefficients for all $i \in [m]$. Without loss of generality
  we may assume\footnote{More formally, this means we can replace each polytope $P_i = \{ (x,y) \in \setR^{d+d_i} \mid A^i x + B^i y \leq b^i\}$ by a polytope $P_i = \{ (x,y) \in \setR^{d + d_{\textrm{aux}}} \mid A^i x + (B^i,\bm{0})y \leq b^i; y_j=0 \; \forall d_i<j\leq d_{\textrm{aux}}\}$ without changing feasibility of the problem.} that $d_i = d_{\textrm{aux}}$ for all $i \in [m]$. We abbreviate $\Delta := \max\{ \|A^i\|_{\infty}, \|B^i\|_{\infty},\|b^i\|_{\infty} : i \in [m]\}$ as the largest coefficient in the representation of any of the polytopes $P_i$. From Lemma~\ref{lem:InfinityNormOfVertex} we know that each point
  $(x,y) \in P_i \cap \setZ^{d + d_{\textrm{aux}}}$ satisfies $\|(x,y)\|_{\infty} \leq (d+d_{\textrm{aux}})! \cdot \Delta^{d + d_{\textrm{aux}}}$. Again, we use the Big-$M$ method to convert the union of polytopes into a single polytope. The choice of $M := (d+d_{\textrm{aux}}+1) \cdot (d+d_{\textrm{aux}})! \cdot \Delta^{d + d_{\textrm{aux}}}$ will be large enough. We set
%  $\beta_i \in \{ 0,1\}^{\log_2(m)}$ be the binary encoding of the number $i$.
  \[
\tilde{P} := \left\{ \big( x, y, z \big) \in \setR^{d + d_{\textrm{aux}} + m} \mid A^ix + B^iy \leq b^i + (1-z_{i}) \cdot M \; \forall i \in [m]; \;\; \sum_{i=1}^m z_i=1; z_i \geq 0 \;\;\forall i \in [m] \right\}
  \]
  and extend the target polyhedron to
  $
   \tilde{Q} := Q \times \setR^{d_{\textrm{aux}}} \times \setR^m = \{ (x,y,z) \mid Ax \leq b \}
  $. We can prove that $\tilde{P}$ indeed encodes the union of $P_i$'s: \\
  {\bf Claim.} \emph{If $e_i$ denotes the $i$th unit vector in $\setR^m$, then one has}
  \begin{equation} \label{eq:IntConeOfIntProj_Ptilde}
    \tilde{P} \cap \setZ^{d + d_{\textrm{aux}} + m} = \bigcup_{i=1}^m \left\{ \big( x, y, e_i \big)  \mid (x,y) \in P_i \cap \setZ^{d + d_{\textrm{aux}}} \right\}.
  \end{equation}
  {\bf Proof of claim.} Let $(x,y,z) \in \tilde{P} \cap \setZ^{d + d_{\textrm{aux}} + m}$. Then by integrality and the constraint $\sum_{i=1}^m z_i=1$ we know that $z=e_i$ for some index $i \in [m]$. Then one has $A^ix+B^iy \leq b^i + \bm{0} \cdot M$ and hence $(x,y) \in P_i \cap \setZ^{d + d_{\textrm{aux}}}$. For the reverse direction, let $(x,y,e_i)$ be contained in the right hand side of \eqref{eq:IntConeOfIntProj_Ptilde}. Clearly the constraint $A^ix + B^iy \leq b^i + (1-z_i)M$ is satisfied. Now consider a constraint for $i' \neq i$. Then one has $A^{i'}x + B^{i'}y \leq d\|A^{i'}\|_{\infty}\|x\|_{\infty} + d_{\textrm{aux}}\|B^{i'}\|_{\infty}\|y\|_{\infty} \leq b^{i'} + (1-z_{i'})M$ using the bound  $\|(x,y)\|_{\infty} \leq (d+d_{\textrm{aux}})! \cdot \Delta^{d + d_{\textrm{aux}}}$ and making use of the choice of $M$. Hence $(x,y,e_i) \in \tilde{P}$ and the claim follows. \qed

  The proven claim implies that the condition $\textrm{int.cone}(X_1 \cup \ldots \cup X_m) \cap Q \neq \emptyset$ holds if and only if 
  \begin{equation} \label{eq:IntConeOfIntProj_equivTest}
\textrm{int.cone}(\tilde{P} \cap \setZ^{d + d_{\textrm{aux}} + m}) \cap \tilde{Q} \neq \emptyset
  \end{equation}
  Hence we can apply Theorem~\ref{thm:MainGeneralTheorem} to decide \eqref{eq:IntConeOfIntProj_equivTest}.
  In the affirmative case, Theorem~\ref{thm:MainGeneralTheorem} will return an integer conic combination $\tilde{\lambda}$ satisfying $\sum_{(x,y,z) \in \tilde{P} \cap \setZ^{\tilde{d}}} \tilde{\lambda}_{(x,y,z)} \cdot (x,y,z) \in \tilde{Q}$ where the support size is bounded by $2^{2\tilde{d}+1}$ with $\tilde{d} := d + d_{\textrm{aux}} + m$. Then for $i \in [m]$ and $x \in X_i$ we set $\lambda_{i,x} := \sum_{y} \tilde{\lambda}_{(x,y,e_i)}$ and return the vector $\lambda = (\lambda_{i,x})_{i \in [m],x \in X_i}$ as the desired integer conic combination.

  It remains to estimate the running time that is required by the algorithm behind Theorem~\ref{thm:MainGeneralTheorem}. From the construction of $\tilde{P} \subseteq \setR^{\tilde{d}}$ we see that $\enc(\tilde{P}) \leq \sum_{i=1}^m O(\enc(P^i)) + O(m\log(M))$. Finally, $\log(M) \leq (d+d_{\textrm{aux}})^{O(1)} \log(\Delta)$ is bounded by a polynomial in $\sum_{i=1}^m \enc(P_i)$. The Theorem then follows. %, while $\enc(\tilde{Q}) = \enc(Q)$.
\end{proof}

\subsection{A general high multiplicity scheduling framework}

%For this sake,
Next, we want to formulate a very general scheduling problem that captures
most scheduling problems.
%We want to introduce a scheduling framework that generalizes 
%most scheduling settings and
Then we will argue that this general problem can be solved in polynomial time
if the underlying parameters are constant. 

For our scheduling framework we assume to have \emph{job types} $j \in [d]$ with $a_j \in \setZ_{\geq 0}$ many copies of type $j$. Moreover, we have \emph{machine types} $i \in [m]$ with 
a maximum number $n_i \in \setZ_{\geq 0}$ of machines that are available. There are two types of cost: 
we incur a cost of $f_i$  for each machine of type $i$ that is used and we pay $c_{ij}$ for each copy of job $j$ that is assigned to a machine of type $i$. 

So far, we have not specified more parameters of the jobs, such as release times, deadlines and running times and we also have not specified whether we allow preemption or
not. For our general framework, the only requirement that we make is that 
the configurations of jobs that can be scheduled on a copy of machine $i$ can
be described as an integer projection
\[
  X_i = \big\{ x \in \setZ_{\geq 0}^d : \exists y \in \setZ^{d_i} : (x,y) \in K_i \big\}
\]
where $K_i$ is a polytope. In other words, if $x \in X_i$ then it has to be possible to 
schedule $x_j$ jobs of type $j$ on a machine of type $i$.
The performance of our method depends on the number of constraints in $K_i$ and 
most crucially on the number of extra variables $d_i$ that are used. 

In particular, this captures the classical settings of preemptive and non-preemptive
scheduling with release times, running times and deadlines. 
Formally, the \emph{general scheduling problem} is to assign the jobs to machines 
in order to minimize the cost, which can be written as an integer linear program
of the form
\begin{eqnarray}
 \min \sum_{i \in [m]} \sum_{x \in X_i} \Big(f_i + \Big(\sum_{j \in [d]} c_{ij}x_j\Big)\Big) \lambda_{i,x} & & \label{eq:GeneralSchedulingI} \\ % + \sum y_i(\sum)
 \sum_{i \in [m]} \sum_{x \in X_i} \lambda_{i,x} x &=& a \label{eq:GeneralSchedulingII} \\ 
 \sum_{x \in X_i} \lambda_{i,x} &\leq& n_i \quad \forall i \in [m] \label{eq:GeneralSchedulingIII} \\ 
  \lambda_{i,x} &\in& \setZ_{\geq 0} \quad \forall i \in [m] \; \forall x \in X_i \label{eq:GeneralSchedulingIV}
\end{eqnarray}
The variable $\lambda_{i,x}$ denotes how many machines of type $i$ should be packed with configuration $x$. The objective function pays a ``fixed cost'' of $f_i$ for each copy of machine $i$ plus a ``variable cost'' of $c_{ij}$ for each single job of type $j$ that is assigned to a machine $i$. The constraint~\eqref{eq:GeneralSchedulingII} forces
that every job copy is assigned and \eqref{eq:GeneralSchedulingIII} guarantees that we do not use more than $n_i$ copies of machine $i$.

\begin{theorem} \label{thm:GeneralScheduling}
Consider a general scheduling problem $\pazocal{I}$ with $d$ many job types, $m$ many machine types
represented by $X_i = \{ x \in \setZ_{\geq 0}^d \mid \exists y \in \setZ^{d_i} : (x,y) \in K_i\}$ with rational polytopes $K_i$.
The general scheduling problem can be solved in time $\textrm{enc}(\pazocal{I})^{2^{O(m + d + d_{\textrm{aux}})}}$
where $d_{\textrm{aux}} := \max_{i \in [m]} d_i$ is the maximum number of auxiliary variables
and $\enc(\pazocal{I})$ is the total encoding length of the vectors $c$, $f$ and $n$ plus the sum of the encoding lengths of the polytopes $K_i$ for $i \in [m]$.
\end{theorem}
\begin{proof}
We may assume that all the cost parameters $c_{ij}$ and $f_i$ are integers.
It suffices to find a solution of total cost $T$ or determine that there is none. 
The claim then follows by performing a binary search on $T$.
Let us write $K_i = \{ (x,y) \mid A^ix + B^iy \leq b^i\}$.
Now, define
\[
  P_i := \left\{ \big((x,c_i^Tx,f_i,e_i),y\big) \in \setR^{d} \times \setR \times \setR \times \setR^m \times \setR^{d_{i}}  \mid (x,y) \in K_i \right\}
\]
    and
 \[
      \tilde{X}_i := \Big\{ \big(x,c_i^Tx,f_i,e_i\big) \in \setZ^{d+2+m} \mid \exists y \in \setZ^{d_i}: \big((x,c_i^Tx,f_i,e_i),y\big) \in P_i\Big\}
\]
where $e_i$ is the $i$th unit vector in $\setR^m$.
Note that any vector in $\tilde{X}_i$ represents a packing of a machine of type $i$ where we have extra coordinates for  variable cost and an extra coordinate for the fixed cost $f_i$ as well as extra coordinates for the vector $e_i$.
Then the integer linear program from \eqref{eq:GeneralSchedulingI}-\eqref{eq:GeneralSchedulingIV} has a solution of value at most $T$, if and only if
\begin{equation} \label{eq:GeneralSchedulingIntConeCondition}
\textrm{int.cone}\big( \tilde{X}_1 \cup \ldots \cup \tilde{X}_m \big) \cap \tilde{Q} \neq \emptyset \quad \textrm{where} \quad \tilde{Q} = \left\{ \big(a,C,F,\mu \big) \mid \begin{array}{c} C+F \leq T \\ \mu \leq n \end{array}  \right\}
\end{equation}
Then we apply Theorem~\ref{thm:IntegerConicCombinationsOfUnionsOfIntegerProjections} to decide the condition in \eqref{eq:GeneralSchedulingIntConeCondition}.
For the running time note that the integer projections $\tilde{X}_i$ live in $\setR^{d + 2 + m}$ and require  $d_i \leq d_{\textrm{aux}}$ many extra integer variables.
\end{proof}

\subsection{Preemptive scheduling}

In the next two subsections, we want to discuss how the polytopes $K_i$ should be chosen
in order to represent preemptive and non-preemptive schedules. 
We begin with \emph{preemptive} scheduling (without migration). 
Note that once the assignment to machines is done, 
the \emph{Earliest-Deadline First policy} (EDF) gives an 
optimum preemptive schedule~\cite{EDFoptimalDertouzos74}. This allows to
check in polynomial time whether a given set of jobs is schedulable one a single machine, even
if the number of job types is not constant. 
\begin{lemma} \label{lem:K-for-preemptive-scheduline}
Given jobs $j \in [d]$, each with running time $p_j \in \setZ_{\geq 0}$, release time $r_j \in \setZ_{\geq 0}$ and 
deadline $d_j \in \setZ_{\geq 0}$. Then the set of multi-sets of such jobs that are EDF-schedulable
on a single machine can be described as $X = \{ x \in \setZ_{\geq 0}^d : x \in K\}$
where $K$ is a polytope with $\enc(K) \leq O(d^3 \cdot \log \Delta)$
and $\Delta := \max\{ p_j,r_j,d_j,2 : j \in [d] \}$. 
% $O(|J|^2)$ many inequalities. 
\end{lemma}
\begin{proof}
Consider a single machine of type $i$ and a vector $x \in \setZ_{\geq0}^d$ of jobs and we wonder how to
determine whether the jobs in $x$ can be scheduled on a single machine, i.e. how to test if the
EDF schedule of a set of jobs containing $x_j$ copies of job $j$ will meet all the deadlines.
If we consider a time interval $[t_1,t_2]$ then it is clear that the total 
running time of all jobs that have both, release time and deadline in $[t_1,t_2]$
cannot be larger than the length $t_2-t_1$, otherwise the schedule must be infeasible. 
In fact, for the EDF-scheduling policy, this is also a sufficient 
condition\footnote{This can be easily derived from Hall's condition
for the existence of perfect matchings in bipartite graphs and the optimality of EDF.}.
Moreover, it is clear that one does not need to consider \emph{all} time intervals, but
just those whose end points lie in the set $T := \{ r_{j},d_{j} \mid j\in[d]\}$ of critical points.
Thus we can define $K$ as the set of vectors $x \in \setR_{\geq0}^d$ such that
\begin{equation} 
  \sum_{\substack{j\in[d]: \\ r_{j},d_{j} \in [t_1,t_2]}} x_jp_{j} \leq t_2-t_1 \; \; \forall t_1,t_2 \in T: t_1\leq t_2. \label{eq:SchedulablePolytopeForEDF}
\end{equation}
Observe that a job vector $x \in \setZ_{\geq0}^d$ can be scheduled 
if and only if $x \in K$. The bound on $\enc(K)$ follows from the fact that $K$ has $d$ variables, $O(d^2)$
constraints and all coefficients are bounded by $\Delta$.
\end{proof}

\subsection{Non-preemptive scheduling}

Next, we consider scheduling without preemption. In contrast to the preemptive case, 
even in the single machine case,  finding a feasible non-preemptive 
schedule is $\mathbf{NP}$-hard~\cite{GareyJohnson79} in general. % (for general $d$, but $a_j = 1$).
%Again, we want to first investigate the case that we have a job vector $x \in \setZ_{\geq0}^d$ and a single machine of some type
%to schedule all these jobs. %, just that we use non-premptive scheduling this time. 
%For the moment, let us abbreviate the release times, deadlines and processing times with $r_j,d_j$ and $p_j$.
However, we will see that for fixed number $d$ of job types, the set of feasible schedules allows
a compact description: %a schedule can be found in polynomial time.
\begin{lemma} \label{lem:K-for-nonpreemptive-scheduline}
Given jobs $j \in [d]$, each with running time $p_j \in \setZ_{\geq 0}$, release time $r_j \in \setZ_{\geq 0}$ and 
deadline $d_j \in \setZ_{\geq 0}$. Then the set of multi-sets of such jobs that are schedulable
non-preemptively on a single machine can be described as 
$X = \{ x \in \setZ_{\geq 0}^d : \exists y \in \setZ^{O(d^2)}: (x,y) \in K\}$
with $\enc(K) \leq O(d^4 \log \Delta)$ where $\Delta := \max\{ p_j,r_j,d_j,2 : j \in [d] \}$. 
%where $\Delta$ is the largest number appearing in the input. 
%where $K$ is a polytope with $O(|J|^2)$ many inequalities.
\end{lemma}
\begin{proof}
  While we index the jobs of the instance as $j \in \{ 1,\ldots,d\}$, it will
  be notationally convinient to set $r_0 := 0, p_0:=1, d_{0} := \infty$, which allows us to use job 0 to represent any  idle time of the machine. 
%We index the jobs as . 
%  For notational convinience we add $\Delta$ many copies of a dummy job with parameters $p_0 = 1$, $r_0 = 0$ 
%and $d_0 := \infty$
%and multiplicity $x_0 := \Delta - \sum_{j=1}^d p_jx_j$. Now 
%so that we can assume there is no idle time in the schedule.

Let $T := \{ r_j,d_j \mid j \in [d] \} = \{ t_1,\ldots,t_{2d}\}$ be the $2d$ critical points, sorted so that 
$t_1 \leq \ldots \leq t_{2d}$. The crucial observation is that in a feasible schedule, we 
can arbitrarily permute jobs that have both start and end time in an interval $[t_k,t_{k+1}]$.
Let us imagine that the schedule is \emph{cyclic} in the sense that the schedule processes 
first some (possibly 0) copies of job type 0, then some (again, possibly 0) jobs of type 1, and so on until type $d$; 
then the scheduler starts again with jobs of type $0$. 
The interval from a job 0 interval to the beginning of the next job 0 interval is called a \emph{cycle}.
Note that the number of copies of job $j$ that are scheduled in a cycle can very well be $0$, so indeed
such a cyclic schedule trivially exists. Moreover, we want to restrict that a job of type $j$ is 
only allowed to be scheduled in a cycle if the \emph{complete} cycle is contained in $[r_j,d_j]$.
But again this restriction is achievable as we could split cycles if needed.

We claim that there exists a schedule with at most $4d$ many cycles. 
%Now consider the schedule with the least number of cycles. 
Following our earlier observation it is clear that whenever 2 cycles are completely contained in some 
interval $[t_{k},t_{k+1}]$ of consecutive points, then we could also join them. 
Hence we have at most $2d$ cycles contained in intervals of the form $[t_k,t_{k+1}]$ plus at most $2d$ cycles
that include a critical point.
%Thus we can assume that the schedule contains exactly $4d$ many cycles (maybe some 
%have length 0).

We introduce an auxiliary variable $y_{jk}$ which tells
us how many copies of job $j$ are processed in the $k$th cycle. Additionally
we have a binary variable $z_{jk}$ telling us whether jobs of type $j$ can be processed in the $k$th cycle. 
Moreover, the $k$th cycle runs in $[\tau_{k-1},\tau_{k}]$ (with $\tau_0 := 0$). Then the polytope $K$ whose integral 
points correspond to feasible schedules can be defined as
\begin{equation} \label{eq:NonPreemptiveSchedulingPolytope}
\begin{array}{rcll}
  x_j  &=&  \sum_{k=1}^{4d} y_{jk} & \forall j \in \{1,\ldots,d\} \\
  \tau_{k} &=& \sum_{\ell \leq k} \sum_{j=0}^d p_jy_{j\ell} & \forall k \in [4d] \\
  y_{jk} &\leq& \Delta \cdot z_{jk} & \forall j \in [d] \; \; \forall k \in [4d] \\
  \tau_{k-1} &\geq& r_j - \Delta (1-z_{jk}) & \forall j \in [d] \; \; \forall k \in [4d] \\
  \tau_{k} &\leq& d_j + \Delta (1-z_{jk}) & \forall j \in [d] \; \; \forall k \in [4d] \\
%  x_0 &=& \Delta - \sum_{j=1}^d x_jp_j \\
  y_{jk},\tau_{k} &\geq& 0 & \forall j \in \{0,\ldots,d\} \; \forall k \in [4d] \\
 % & & &  \\
  z_{jk} &\in& [0,1] & \forall j \in [d] \; \forall k \in [4d].
\end{array}
\end{equation}
For example if $y_{jk} > 0$ then this forces that $z_{jk} = 1$ and hence $r_j \leq \tau_{k-1} \leq \tau_k \leq d_j$.
% Note that this system has $O(d^2)$ many variables and inequalities.
A vector $x \in \setZ_{\geq0}^d$ can be non-preemptively scheduled if and only if
there are integral $\tau,y,z$ such that $(x,\tau,y,z) \in K$. Since $K$ has $O(d^2)$ variables and $O(d^2)$
constraints, we obtain the bound $\enc(K) \leq O(d^4 \log \Delta)$.
\end{proof}
One might be tempted to wonder whether the number of variables could be reduced
at the expense of more constraints, which might still improve the running time.
But for non-preemptive scheduling we run into the problem that the 
set of vectors $x$ that can be scheduled on a single machine is not closed under taking
convex combinations\footnote{A simple example is the following: consider a set of $d=3$ job types with $\{ (r_j,d_j,p_j) \mid j=1,2,3\} = \{ (0,300,150), (100,102,1), (200,202,1)\}$. The vectors $x' = (2,0,0)$ and $x'' = (0,2,2)$ can both be scheduled in a non-preemptive way. 
But the convex combination $\frac{1}{2}(x' + x'') = (1,1,1)$ cannot be scheduled.}. 
In fact, some additional variables are necessary to write those vectors
as integer projection of a convex set.

\subsection{A Scheduling Application}

As said earlier, many scheduling problems, that involve a constant number of job types and machine types
can be handled by our framework. For the sake of demonstration, let us describe one natural setting:
\begin{corollary}
Given job types $j \in [d]$ and machine types $i \in [m]$ where job $j$ has machine-dependent release time $r_{ij} \in \setZ_{\geq 0}$, deadline  $d_{ij} \in \setZ_{\geq 0}$
and running time $p_{ij} \in \setZ_{\geq 0}$ on a machine of type $i \in [m]$.
Moreover we have $a_j \in \setZ_{\geq 0}$ copies of job type $j$ and using a copy
of a machine of type $i \in [m]$ incurs a cost of $c_i \in \setZ_{\geq 0}$.
Then one can find an optimum assignment of jobs to machines minimizing the total machine cost 
under preemptive [non-preemptive, resp.] scheduling in time $(\log \Delta)^{2^{O(d+m)}}$ [$(\log \Delta)^{2^{O(d^2+ m)}}$, resp.],
where $\Delta := \max\{ \|a\|_{\infty},\|r\|_{\infty},\|d\|_{\infty},\|p\|_{\infty},4\}$.
\end{corollary}
\begin{proof}
For the preemptive case, choose $K_i$ as in Lemma~\ref{lem:K-for-preemptive-scheduline} with $d_{\textrm{aux}} = 0$. In the preemptive
case, we choose $K_i$ as in Lemma~\ref{lem:K-for-nonpreemptive-scheduline}
with $d_{\textrm{aux}} = O(d^2)$ auxiliary variables.
\end{proof}

It is not difficult to incorporate other objective functions. For example if the goal is
to \emph{minimize the makespan}, then one can perform a binary search on the target makespan $D$. 
In each search step one sets the deadlines $d_{ij}$ to the current value of $D$ and then runs Theorem~\ref{thm:GeneralScheduling} to check feasibility. 
As another example, consider the case that the objective function is to \emph{minimize the tardy jobs}. 
One can set the machine cost to 0 and add a dummy machine with infinite resources that charges
a cost of 1 per scheduled job. Then the minimum cost solution will try do maximize the jobs that 
are assigned to the regular machines. 
This concludes the discussion on scheduling. 

\section{Finding integer conic combinations for unbounded polyhedra\label{sec:IntConicComForUnboundedPolyhedra}}

So far we have only considered integer combinations $\textrm{int.cone}(P \cap \setZ^d)$
where $P$ was a \emph{bounded} polyhedron. Now we will discuss how that boundedness assumption can be removed.
We will see that $P$ can be decomposed so that the test $\textrm{int.cone}(P \cap \setZ^d) \cap Q \neq \emptyset$ is equivalent to a test
where points come from the union of two bounded polytopes. Such a test can then be decided
by our Theorem~\ref{thm:IntegerConicCombinationsOfUnionsOfIntegerProjections}.
%Recall that we already know that $\textrm{int.cone}( (P_1 \cup \ldots \cup P_m) \cap \setZ^d) \cap Q \neq \emptyset$
%can be tested in polynomial time if both $d$ and $m$ are constants and $P_1,\ldots,P_m \subseteq \setR^d$ are bounded polytopes. We will make use of this fact to compute integer conic combinations for unbounded polyhedra.
\begin{theorem}
Let $P,Q \subseteq \setR^d$ be rational polyhedra. Then the condition $\textrm{int.cone}(P \cap \setZ^d) \cap Q \neq \emptyset$ can be decided in time $\textrm{enc}(P)^{2^{O(d)}} \cdot \textrm{enc}(Q)^{O(1)}$. In the affirmative case, a vector $(\lambda_x)_{x \in P \cap \setZ^d}$ of support  $2^{O(d)}$ with $\sum_{x \in P \cap \setZ^d} \lambda_x x \in Q$ can be computed in the same time.
\end{theorem}
%\begin{proof}
%  Let $P = \{ x \in \setR^d \mid Ax \leq b\}$. Let $M := d!\Delta^d$. Then $P \cap [-M,M]^d$ contains all the extreme points of $P$. Consider the characteristic cone
%  $C := \{ x \in \setR^d \mid Ax \leq 0\}$. Then there are integer vectors $a_I \in \setZ^d$ so that $C = cone\{ a_1,\ldots,a_N\}$. By the same argument $\|a_i\|_{\infty} \leq M$. $int.cone(P \cap \setZ^d)$. Consider the zonotope $R$
%\end{proof}

\begin{proof}
  We write $P = \{ x \in \setR^d \mid Ax \leq b\}$.
  Again we may also assume that $A$ and $b$ have only integral entries and set $\Delta := \max\{ \|A\|_{\infty},\|b\|_{\infty},2\}$.
  The \emph{Minkowski-Weil Theorem} (see e.g. \cite{TheoryOfLPandIP-Schrijver1999}) tells us that 
  one can decompose $P = S + C$ where $S \subseteq \setR^d$ is a polytope and $C = \{ x \in \setR^d \mid Ax \leq \bm{0}\}$
  is the \emph{characteristic cone} of $P$. As $P$ is rational, a valid choice\footnote{If $P$ has vertices then by Lemma~\ref{lem:InfinityNormOfVertex}, $S$ contains all the vertices and the result is immediate. However, it is possible that $P$  contains lines. Hence, we consider a maximal index set $I \subseteq [d]$ so that the restriction $P' = \{ x \in P \mid x_i = 0 \; \forall i \in I\}$ still satisfies $P' + C = P$. Then $P'$ does not contain a line while it intersects every minimal face of $P$.  The minimal faces of $P'$ are vertices and  every  $x \in \textrm{vert}(P')$   has $\|x\|_{\infty} \leq d!\Delta^d = M$. Then $\textrm{vert}(P') \subseteq S$ which concludes the argument.}  is $S = P \cap [-M,M]^d$  where  $M := d!\Delta^d$.
  Also the cone $C$ is rational and hence $C = \textrm{cone}\{ a_1,\ldots,a_N\}$ for some integer vectors $a_1,\ldots,a_N \in \setZ^d$.
  Similar to before one can argue that coefficients of size $\|a_i\|_{\infty} \leq M$ suffice\footnote{Consider the polytope $C_{\textrm{bounded}} := C \cap [-1,1]^d$. We know that the vertices of $C_{\textrm{bounded}}$ span the cone, i.e. $\textrm{cone}(\textrm{vert}(C_{\textrm{bounded}})) = C$. While the vertices of $C_{\textrm{bounded}}$ are in general not integral, we can scale them to become integral with bounded coefficients. More precisely by Cramer's rule, every $x \in \textrm{vert}(C_{\textrm{bounded}})$ is of the form $x = (\frac{\det(R_1)}{\det(T)},\ldots,\frac{\det(R_d)}{\det(T)})$ where $R_1,\ldots,R_d,T$ are $d \times d$ matrices filled with entries from $\{ -\Delta,\ldots,\Delta\}$. Then $\det(T) \cdot x \in \setZ^d$ and moreover $\|\det(T) \cdot x\|_{\infty} = \max_{i \in [d]} |\det(R_i)| \leq d!\Delta^d = M$.}  for all $i \in [N]$.
  The next step is to restrict the cone to a bounded region by setting $R := C \cap [-dM,dM]^{d}$.
  In the following we abbreviate $\textrm{int.cone}_{+}(X) := \{ \sum_{x \in X} \lambda_x x \mid \lambda_x \in \setZ_{\geq 0}\textrm{ and }\lambda \neq \bm{0} \}$ as the integer conic combinations of a set $X$ where the combination with $\lambda = \bm{0}$ is excluded. Now we can analyze the decomposition: \\
  {\bf Claim I.} \emph{One has $\textrm{int.cone}_+(P \cap \setZ^d) = \textrm{int.cone}_+((S +R)\cap \setZ^d) + \textrm{int.cone}(R \cap \setZ^d)$}. \\
  {\bf Proof of Claim I.}
  For the direction ``$\subseteq$'', we will show that $P \cap \setZ^d \subseteq \textrm{int.cone}_+((S +R)\cap \setZ^d) + \textrm{int.cone}(R \cap \setZ^d)$ holds. Consider a point $x \in P \cap \setZ^d$
  and write it as $x = s + c$ where $s \in S$ and $c \in C$ (here neither $s$ nor $c$ is guaranteed to be integral). Then by the definition of the cone $C$, we can write $c = \sum_{i=1}^N \mu_ia_i$ for some coefficients $\mu_i \in \setR_{\geq 0}$ where
  by \emph{Carath\'eodory's Theorem} (again see e.g. \cite{TheoryOfLPandIP-Schrijver1999})  we may assume that $|\textrm{supp}(\mu)| \leq d$. 
  We write
  \[
 c = \sum_{i=1}^N \mu_ia_i = \underbrace{\sum_{i=1}^N (\mu_i - \lfloor \mu_i \rfloor) a_i}_{=: r} + \underbrace{\sum_{i=1} ^N\lfloor \mu_i \rfloor a_i}_{=: c_{\textrm{int}}}
  \]
  where $c_{\textrm{int}} \in \textrm{int.cone}\{ a_1,\ldots,a_N\} \subseteq \textrm{int.cone}(R \cap \setZ^d)$. Moreover we know that $r \in R$
  because $\sum_{i=1}^N (\mu_i-\lfloor \mu_i \rfloor) \leq d$ and $\|a_i\|_{\infty} \leq M$ for $i \in [N]$. Then $x = (s + r) + c_{\textrm{int}}$ where $x$ and $c_{\textrm{int}}$ are integer vectors and so $(s + r) \in (S+R) \cap \setZ^d$.
  
  For the direction ``$\supseteq$'', take integer conic combinations $\sum_{x \in (S+R) \cap \setZ^d}\lambda_x x$ with $\lambda_x \in \setZ_{\geq 0}$ and $\lambda \neq \bm{0}$ and $\sum_{y \in R \cap \setZ^d} \mu_y y$ with $\mu_y \in \setZ^d_{\geq 0}$. Fix a vector $x^* \in (S+R) \cap \setZ^d$ with $\lambda_{x^*} \geq 1$ (this is where we need that $\lambda \neq \bm{0}$!). Set $x^{**} := x^* + \sum_{y \in R \cap \setZ^d} \mu_yy$ and note that $x^{**} \in P \cap \setZ^d$. Define a new integer conic combination $\tilde{\lambda}$ by
  \[
    \tilde{\lambda}_{x^{**}} := \lambda_{x^{**}}+1, \quad \tilde{\lambda}_{x^*} := \lambda_{x^*}-1, \quad \tilde{\lambda}_{x} := \lambda_x \; \forall x \in ((S+R) \cap \setZ^d) \setminus \{ x^*,x^{**}\}, \quad \tilde{\lambda}_{x} := 0 \; \textrm{otherwise}\footnote{In the fringe case where $x^{**} = x^{*}$, set $\tilde{\lambda}_{x}:=\lambda_{x}$ for all $x$.}. %for    x \in (P \setminus (S+R) \setmiuns \{ x^*,x^**}) \cap \setZ^{d} 
  \]
  In words, we have moved all the weight from points in $R \cap \setZ^d$ to a single copy of a single point in $P \cap \setZ^d$.  Then  $\sum_{x \in P \cap \setZ^d}\tilde{\lambda}_xx=\sum_{x \in (S+R) \cap \setZ^d} \lambda_xx + \sum_{y \in R \cap \setZ^d} \mu_yy$ which shows Claim I.
  \qed

  It appears that the overall Theorem almost follows by applying Theorem~\ref{thm:IntegerConicCombinationsOfUnionsOfIntegerProjections}, if there was not the slight issue that Claim I contains the expression $\textrm{int.cone}_+(\ldots)$ which forbids the all-$0$ coefficients.
  However, this issue can be fixed by extending the dimension of the polytopes by 1.
  In particular we will use $P \times \{ 1\} = \{ { x \choose 1}  \mid x \in P\}$ which is a $(d+1)$-dimensional polyhedron.
    Note that if $\bm{0} \in Q$, then trivially $\textrm{int.cone}(P \cap \setZ^d) \cap Q \neq \emptyset$.
    Hence we may assume from now on that $\bm{0} \not\in Q$. Our goal is to reformulate the target condition
    \begin{equation}
    \textrm{int.cone}(P \cap \setZ^d) \cap Q \neq \emptyset.  \label{eq:IntConicCombUnbounded-EquivalenceI}
  \end{equation}
  First note that as $\bm{0} \notin Q$, the coefficient vector $\lambda = \bm{0}$ is useless for \eqref{eq:IntConicCombUnbounded-EquivalenceI} and so \eqref{eq:IntConicCombUnbounded-EquivalenceI} is equivalent to
      \begin{equation}
    \textrm{int.cone}_+\big( (P \times \{ 1\}) \cap \setZ^{d+1}\big) \cap \big(Q \times [1,\infty)\big)  \label{eq:IntConicCombUnbounded-EquivalenceII} \neq \emptyset.
  \end{equation}
  Then by Claim I, \eqref{eq:IntConicCombUnbounded-EquivalenceII} is equivalent to
  \begin{equation}
\Big(\textrm{int.cone}_+\big(((S+R) \times \{ 1\}) \cap \setZ^{d+1}\big) + \textrm{int.cone}\big((R \times \{ 0\}) \cap \setZ^{d+1}\big)\Big) \cap \big(Q \times [1,\infty)\big) \neq \emptyset \label{eq:IntConicCombUnbounded-EquivalenceIII}
\end{equation}
Then \eqref{eq:IntConicCombUnbounded-EquivalenceIII} is equivalent to
\begin{equation}
  \textrm{int.cone}\Big(\big(((S+R) \times \{ 1\}) \cap \setZ^{d+1}\big) \cup \big((R \times \{ 0\}) \cap \setZ^{d+1}\big)\Big) \cap \big(Q \times [1,\infty)\big) \neq \emptyset, \label{eq:IntConicCombUnbounded-EquivalenceVI}
\end{equation}
where the extra coordinate enforces that at least one point from $((S+R) \times \{ 1\}) \cap \setZ^{d+1}$ needs to have positive weight in any valid integer conic combination.
Now we call Theorem~\ref{thm:IntegerConicCombinationsOfUnionsOfIntegerProjections} to decide condition~\eqref{eq:IntConicCombUnbounded-EquivalenceVI} using that $(S+R) \times \{ 1\}$ and $R \times \{ 0\}$ are bounded.
Note that in the affirmative case, an integer conic combination  satisfying \eqref{eq:IntConicCombUnbounded-EquivalenceVI} can be transformed into one using points of $P \cap \setZ^d$ using the argument from Claim I.

For the running time, note that Theorem~\ref{thm:IntegerConicCombinationsOfUnionsOfIntegerProjections} is called for $(d+1)$-dimensional polytopes whose encoding length is bounded polynomially in $\textrm{enc}(P)$.
 % Any point $y \in X$ can be written as $y = \sum_{x \in P \cap \setZ^d} \lambda_x \begin{pmatrix} x \\ 1 \end{pmatrix} \in X$. Then 
\end{proof}

\section{The Eisenbrand-Shmonin Theorem is tight\label{sec:EisenbrandShmoninIsTight}}

In this section, we want to describe an example that shows that the Eisenbrand-Shmonin result described in Lemma \ref{lem:ExistenceOfIntegralSolWithSupport2d}
is tight up to a factor of 2. To the best of our knowledge, this was not known before. 
%The Eisenbrand-Shmonin Theorem says that any integer conic 
%For points $X \subseteq \setR^d$, let $\textrm{int.cone}(X) = \{ \sum_{x \in X} \lambda_x x \mid \lambda_x \in \setZ_{\geq 0} \; \forall x \in X\}$ be the set of \emph{integer conic combinations}.
%Recall the following theorem of Eisenbrand \& Shmonin:
%\begin{theorem}
%Let $P \subseteq \setR^d$ be a convex set and $y \in \textrm{int.cone}(P \cap \setZ^d)$. 
%Then there is an integer conic combination $\lambda \in \setZ_{\geq 0}^{P \cap \setZ^d}$ with $y = \sum_{x \in P \cap \setZ^d} \lambda_x %x$  with $|\supp(\lambda)| \leq 2^d$.
%\end{theorem}
%We want to argue that this is essentially tight. 
%Let us define $\textrm{minsup}(y,P) = \min\{ |\supp(\lambda)| \mid \lambda \in \setZ_{\geq 0}^{P \cap \setZ^d}, \sum_{x \in P \cap \setZ^d} \lambda_x x = y\}$ as the minimal support 
%of an integer conic combination
%defining $y$ (to be well defined, say $\textrm{minsup}(y,P) = 0$ if there is no combination possible at all).
%Then the Eisenbrand-Shmonin Theorem can be reinterpreted as saying that
%$\textrm{minsup}(y,P) \leq 2^d$ for all $y \in \textrm{int.cone}(P \cap \setZ^d)$ were $P$ is convex.
Fix a dimension $d\geq 2$ and let $k := 2^{d-1}$. Let's define a set $X\subset \setZ^{d}$ of $k$ points as
\[
\hspace{-0.3cm} \left\{ \big(1+x_1,\ldots,1+x_{d-1},(4k)^{1+\sum_{i=1}^{d-1} 2^{i-1} x_i}\big) \mid x_i \in \{ 0,1\} \;\; \forall i \in [d-1] \right\}.
\]
For example, for $d=3$ we obtain
\[
  \begin{pmatrix} 1 \\ 1 \\ (4k) \end{pmatrix},  
\begin{pmatrix} 2 \\ 1 \\ (4k)^2 \end{pmatrix}, 
 \begin{pmatrix} 1 \\ 2 \\ (4k)^3 \end{pmatrix},  \begin{pmatrix} 2 \\ 2 \\ (4k)^4 \end{pmatrix}.
\]
We sort $X = \{ a_1,\ldots,a_k \}$ according to their length, i.e.  $\| a_i\|_{\infty} = (4k)^i$ and define $P := \conv{X}$. %Due to the first $d-1$ coordinates, 

\begin{lemma}
The integer conic combination $y := a_1 + \ldots + a_k$ is unique, thus there are $2^{d-1}$ 
points necessary to obtain $y$ as an integer conic combination of points in $P \cap \setZ^d$. %have $ \textrm{minsup}(  y,  P)  = k$.
\end{lemma}
\begin{proof}
First, we observe that there is
no other integer point in the convex hull $P$ because of the first $d-1$ coordinates of the $a_i$'s. 
In other words $P \cap \setZ^d = X$. We want to argue that due to the enormous growth of the 
last coordinate, each $a_i$ has to be used exactly once in order to obtain $y$.
%$y = a_1 + \ldots + a_k$ is the only possible integer conic combination. 
For this sake, consider any integer conic combination
\[
   y = \sum_{i=1}^k \lambda_i a_i.
\]
Note that $y_j = 3\cdot  2^{d-2}$ for $j \in \{ 1,\ldots,d-1\}$, thus $0 \leq \lambda_i \leq 3\cdot 2^{d-2} < 2k$
for $i=1,\ldots,k$. 
We want to argue that $\lambda_1 = \ldots = \lambda_k = 1$ is the only possibility.
Suppose this is not the case and let  $i^*$ be the largest index with $\lambda_{i^*} \neq 1$. We will see that if $\lambda_{i^*}=0$, then
the combined vector is too short and if $\lambda_{i^*} \geq 2$, then it is too long. 
More formally, we inspect the difference
\begin{eqnarray}
   \Big\| \sum_{i=1}^k (\lambda_i-1)a_i \Big\|_\infty  % = \sum_{i=i^*}^k (\lambda_i-1)a_i
%\geq |\lambda_{i^*}-1| \cdot \|a_{i^*}\|_2 - \sum_{i=i^*}^k \underbrace{\lambda_i}_{\leq 2^n} \|a_i\|_{2} > ( |\lambda_{i^*}-1| - \frac{1}{2}) \cdot \|a_{i^*}\|_2 
&\geq& |\lambda_{i^*}-1| \cdot \|a_{i^*}\|_{\infty} - \sum_{i=1}^{i^*-1} \underbrace{\lambda_i}_{\leq 2k} \cdot \underbrace{\|a_i\|_{\infty}}_{= (4k)^i} \label{eq:ES-is-tightI} \\
&\geq & |\lambda_{i^*}-1| \cdot (4k)^{i^*} - \sum_{i=1}^{i^*-1} 2k \cdot (4k)^{i} \label{eq:ES-is-tightII} \\
& \geq & (4k)^{i^*} \cdot \Big(|\lambda_{i^*}-1| - \frac{1}{2} \underbrace{\Big(\sum_{j\geq 0} (4k)^{-j}\Big)}_{\leq 3/2} \Big) \label{eq:ES-is-tightIII} \\
&\geq&  \underbrace{\Big(|\lambda_{i^*}-1| - \frac{3}{4}\Big)}_{\geq 1/4} \cdot (4k)^{i^*} > 0 \label{eq:ES-is-tightIV} 
\end{eqnarray}
In \eqref{eq:ES-is-tightI} we use the reverse triangle inequality and the fact
that $|\lambda_i-1|=0$ for $i>i^*$. In \eqref{eq:ES-is-tightII} we make use of
$\lambda_i \leq 2k$ and $\|a_i\|_{\infty} = (4k)^i$.
To bound the sum in \eqref{eq:ES-is-tightIII} we use that $2k\geq 4$. Finally in
\eqref{eq:ES-is-tightIV} we reach the conclusion of
$\big\| (\sum_{i=1}^k \lambda_ia_i)-y \big\|_\infty > 0$ which is a contradiction
to the assumption that $\lambda$ is a valid integer conic combination for $y$ .
%But if $\lambda_{i^*} \neq 1$, then the length of this vector is
%strictly larger than $0$, hence $\lambda$ does not correspond to a valid integer conic combination for $y$. 
\end{proof}
Note that $\|a_k\|_{\infty} = (4k)^k = 2^{\Theta(d2^{d})}$, hence our construction 
uses numbers that are doubly-exponentional in $d$.
An argument of~\cite{IntegerConesEisenbrandShmoninORL06} based on the pigeonhole principle shows that for a set $X \subseteq \setZ^d$, every point in $\textrm{int.cone}(X)$ admits an integer conic combination 
whose support size is bounded by
$O(d \cdot \log (d M))$, where $M := \max\{ \|x\|_{\infty} \mid x \in X\}$. 
In other words, any set of integral vectors $X$ where some vector $\textrm{int.cone}(X)$
requires a support of size $\Omega(2^d)$ must have $M \geq 2^{\Omega(2^d)}$.
%contain integer points with coordinates as 
%large as $2^{\Omega(2^d)}$, so the 
Hence the doubly exponentially large numbers are indeed necessary.

\paragraph{Follow-up work.}

After publication of the conference version of this paper, Jansen and Klein~\cite{StructureOfIntegerCone-JansenKlein-SODA2017} proved a modification of the Structure Theorem~\ref{thm:StructureTheorem} where the points 
$X$ are simply the vertices of $\textrm{conv}(P \cap \setZ^d)$. Interestingly, in order for that Theorem to 
be true, the used number of points from $(P \cap \setZ^d) \setminus X$ counted with multiplicity 
has to be increased from $2^{\Theta(d)}$ to $2^{2^{\Theta(d)}}$. Their main algorithmic result is the
following:
\begin{theorem}[Jansen-Klein~\cite{StructureOfIntegerCone-JansenKlein-SODA2017}]
Given rational polytopes $P, Q \subseteq \setR^d$, one can find a vector $y \in  \textrm{int.cone}(P \cap \setZ^d) \cap Q$ and
a vector $\lambda \in \setZ_{\geq0}^{P \cap \setZ^d}$ in time
$|\textrm{vert}(P_I)|^{2^{O(d)}} \cdot \textrm{enc}(P)^{O(1)} \cdot \textrm{enc}(Q)^{O(1)}$
such that $y = \sum_{x \in P \cap \setZ^d} \lambda_x x$, or decide that no such $y$ exists.
%Moreover, the support of $\lambda$ is always bounded by $2^{2d+1}$. %\marginpar{If Eisenbrand-Shmonin is polynomial, then we can put $2^d$ here.}
\end{theorem}
Here $\textrm{vert}(P_I)$ are the extreme points of the convex hull of integers points in $P$. Plugging in the worst case bounds of Cook et al.~\cite{IntegerPointsInPolyhedra-CookHartmannKannanMcDiarmid-Combinatorica92, ComplexityOfIntegerHull-TechReport-Hartmann1988} does not improve over our running time of $O(\log \Delta)^{2^{O(d)}}$ even
for bin packing in $d$ dimensions. However, one might imagine applications
where the polytopes in question are guaranteed to have rather few vertices
so that the approach of \cite{StructureOfIntegerCone-JansenKlein-SODA2017} dominates.

\paragraph{Open problems.}

A natural open problem that arises from this work is whether the double exponential running time is necessary. We pose it as an open problem whether bin packing in $d$
dimensions can be solved in time $f(d) \cdot O(\log(\Delta))^{O(1)}$, where $f(d)$
is an arbitrary function depending on the dimension. Phrased differently, we ask
whether bin packing is \emph{fixed-parameter trackable} with the dimension as parameter.

We would like to point out that there are problems in high multiplicity scheduling
that remain unsolved. For example, imagine that we are given $d$ different rectangles
with dimensions $w_i \times h_i$ and multiplicity $n_i$ as well as a $W \times H$ sized bin. The question is whether all the rectangles can be packed into this bin (without rotating the rectangles). Schiermeyer~\cite{PackingRectangles-Schiermeyer-ESA94} conjectures that this problem is
solvable in polynomial time if $d$ is fixed. Our framework
does not seem to apply (at least not without structural insights into the optimum rectangle packing).

%Moreover, Bin Packing admits a very strong column based linear program
%$OPT = \lceil OPT_f \rceil + 1$.

%Note that overall, this does not improve the asymptotic worst-case running time of 
%$O(\log \Delta)^{2^{O(d)}}$ for Bin Packing in $d$ dimensions.

\bibliographystyle{alpha}
\bibliography{BinPackingInFixedDim}

\end{document}